\pdfoutput=1
\documentclass[11pt,a4paper]{article}

\usepackage[T1]{fontenc}
\usepackage[utf8]{inputenc}
\usepackage{lmodern}
\usepackage{microtype}
\usepackage[a4paper,margin=1in]{geometry}

\usepackage{amsmath,amssymb,amsthm,mathtools}

\usepackage{graphicx}
\usepackage{booktabs}
\usepackage{tabularx}
\usepackage{enumitem}
\usepackage{algorithm}
\usepackage{algpseudocode}

\newcolumntype{C}{>{\centering\arraybackslash}X}

\usepackage{authblk}
\usepackage{orcidlink}
\usepackage{xcolor}
\usepackage{hyperref}

\hypersetup{
  colorlinks=true,
  linkcolor=blue,
  citecolor=blue,
  urlcolor=blue,
  pdftitle={
    Ultrametric Convergence of Guarded Automata
    and Applications to Structural Input Validation
  },
  pdfauthor={
    Atanas Ilchev, Hristo Kiskinov,
    George Pashev, Boyan Zlatanov
  }
}

\newcommand{\Lang}{\mathcal{L}}
\newcommand{\one}{\mathbf{1}}

\newtheorem{Theorem}{Theorem}[section]
\newtheorem{Lemma}[Theorem]{Lemma}
\newtheorem{Corollary}[Theorem]{Corollary}
\newtheorem{Proposition}[Theorem]{Proposition}

\theoremstyle{definition}
\newtheorem{Definition}[Theorem]{Definition}
\newtheorem{Example}[Theorem]{Example}

\theoremstyle{remark}
\newtheorem{Remark}[Theorem]{Remark}

\newcommand{\authorcontributions}[1]{%
  \section*{Author Contributions}
  #1
}

\newcommand{\funding}[1]{%
  \section*{Funding}
  #1
}

\newcommand{\dataavailability}[1]{%
  \section*{Data Availability Statement}
  #1
}

\newcommand{\acknowledgments}[1]{%
  \if\relax\detokenize{#1}\relax
  \else
    \section*{Acknowledgments}
    #1
  \fi
}

\newcommand{\conflictsofinterest}[1]{%
  \section*{Conflicts of Interest}
  #1
}

\newcommand{\reftitle}[1]{}

\title{%
  Ultrametric Convergence of Guarded Automata
  and Applications to Structural Input Validation
}

\author[1]{%
  Atanas Ilchev\,\orcidlink{0000-0002-9690-2863}%
  \thanks{%
    Corresponding author:
    \href{mailto:atanasilchev@uni-plovdiv.bg}
    {atanasilchev@uni-plovdiv.bg}
  }%
}

\author[1]{%
  Hristo Kiskinov\,\orcidlink{0000-0002-3912-4866}%
}

\author[2]{%
  George Pashev\,\orcidlink{0000-0001-8148-4737}%
}

\author[1]{%
  Boyan Zlatanov\,\orcidlink{0000-0002-5857-9372}%
}

\affil[1]{%
  Department of Mathematical Analysis,
  Faculty of Mathematics and Informatics,
  Paisii Hilendarski University of Plovdiv,
  4000 Plovdiv, Bulgaria
}

\affil[2]{%
  Department of Computer Informatics,
  Faculty of Mathematics and Informatics,
  Paisii Hilendarski University of Plovdiv,
  4000 Plovdiv, Bulgaria
}

\date{July 2026}

\begin{document}

\maketitle

\begin{abstract}
We equip language-equivalence classes of deterministic finite automata
with a distinguishing-word ultrametric and identify the resulting space
isometrically with the regular languages. This space is incomplete,
while its metric completion is naturally identified with the complete
ultrametric space of all formal languages. Guarded language operators
induce contractions on the automaton space, and their Picard iterates
converge in the completion to the unique language fixed point, which is
represented by a finite automaton exactly when it is regular.
Motivated by structural input validation, we use this framework to
construct depth-capped deterministic finite automata with certified
finite-depth correctness. These automata provide efficient pre-filters
for nested input structures, such as parenthesised SQL parameters,
while avoiding the backtracking risks of regular-expression engines
and the runtime overhead of full context-free parsers. We also outline
a practical WAF pipeline combining learned grammar models,
finite-state construction, and \(O(1)\)-memory runtime validation.
\end{abstract}


\medskip


\section{Introduction}

Metric and topological methods provide a natural framework for comparing formal languages through the finite words on which they agree or differ. The corresponding length-based ultrametric and Cantor-type topology were studied in the early literature on language spaces~\cite{Vianu1977,FulopKephart2015}. Related metric and fixed-point methods have also been used in the semantics of recursive definitions and domain equations~\cite{deBakkerZucker1982,AmericaRutten1989}, while shortest-distinguishing-input distances appear naturally in quantitative approaches to automata and regular expressions~\cite{BonchiKonigPetrisan2018,Rozowski2024}.

In recent work, guarded language operators \(T_{S,G}(L)=S\cup\bigcup_{r=1}^{p}u_rLv_r\) were shown to be contractions on the complete length-based ultrametric space of formal languages whenever every recursive occurrence is protected by a guard of positive length~\cite{HristovIlchevKulinaZlatanov2026}. This yields a unique fixed point, quantitative Picard convergence, and finite-depth stabilization. The framework was subsequently extended to positive--negative guarded systems involving both direct and complemented language dependencies~\cite{AjetiHristovIlchevZlatanov2026}.

The present paper develops the corresponding automata-theoretic framework. We equip language-equivalence classes of deterministic finite automata with the distinguishing-word ultrametric and identify this space isometrically with the regular languages. We prove that it is dense but not complete and that its metric completion is naturally the complete space of all formal languages. Guarded operators with regular seeds induce contractions on the automaton quotient space, and their Picard iterates form Cauchy sequences converging in the completion to the unique language fixed point.

A fixed point exists inside the automaton space precisely when the limiting language is regular. If the fixed point is nonregular, every finite Picard iterate is nevertheless represented by a deterministic finite automaton and agrees with the limiting language up to an explicitly certified word length. We also consider systems of guarded operators, two-sided approximation, sound iteration from below, state-complexity estimates, and depth-capped finite-state approximations required for the validation application developed later in the paper.

The reliable validation of structured input data remains a fundamental challenge in cybersecurity, particularly in the prevention of injection attacks such as SQL injection (SQLi).
Modern industrial systems typically employ Web Application Firewalls (WAFs) as a first line of defense, or analyze network flow data to detect SQLi payloads~\cite{CrespoMartinez2023}.
However, the underlying mechanisms of these filters present a difficult trade-off between expressiveness, performance, and security.
Currently, the industry relies on two primary approaches:
\begin{enumerate}
    \item \textbf{Regex-based Heuristics:} Most WAFs utilise Perl-Compatible Regular Expressions (PCRE) to detect malicious payloads.
While expressive, PCRE engines employ backtracking algorithms to simulate stack-like memory.
This makes them highly susceptible to Regular Expression Denial of Service (ReDoS) attacks~\cite{CrosbyWallach2003,StaicuPradel2018}, where an attacker crafts a specific payload (e.g., asymmetric nested brackets) that forces the engine into exponential evaluation time.
Furthermore, standard regular expressions are mathematically incapable of accurately parsing arbitrarily nested structures.
    \item \textbf{Full Context-Free Grammar (CFG) Parsers:} To address the shortcomings of Regex, some advanced systems attempt to parse the entire input using Pushdown Automata (PDA) or abstract syntax trees (ASTs), such as through parse tree validation~\cite{Buehrer2005} or runtime monitoring against grammar models~\cite{HalfondOrso2005}.
While structurally accurate, CFG parsers require dynamic memory allocation (stack management) at runtime, leading to $O(N)$ spatial complexity and significant latency bottlenecks, rendering them unsuitable for high-throughput environments where read-heavy operations dominate.
\end{enumerate}

\textbf{Motivation and Approach.} This paper bridges the gap between these two extremes.
Our goal is to achieve the structural awareness of a CFG parser while maintaining the strict $O(N)$ execution time and $O(1)$ memory footprint of a strictly Deterministic Finite Automaton (DFA).
To achieve this, we shift the computational burden from \emph{runtime} to \emph{compile-time}.
We model nested structures using "guarded language operators" and analyze them as contractions within a complete ultrametric space of formal languages.
By deliberately capping the nesting depth (a pragmatic choice, as legitimate application parameters rarely exceed 3 or 4 levels of nesting), we formally guarantee that the context-free language degenerates into a regular language. This concept shares similar motivations with the regular approximation of context-free languages through transformations~\cite{MohriNederhof2001}.
This allows us to pre-compute a minimal DFA that acts as an ultra-fast, ReDoS-immune pre-filter, guaranteeing the rejection of overly complex or malicious payloads without ever relying on runtime backtracking.

The remainder of the paper is organised as follows. Section~2 collects the necessary preliminaries on formal languages, the length-based ultrametric, deterministic finite automata, and guarded language operators. Section~3 presents the principal theoretical results for automata. Section~3.1 introduces the distinguishing-word ultrametric on language-equivalence classes of deterministic finite automata, identifies this space isometrically with the regular languages, determines its metric completion, and proves the contraction, Cauchy convergence, fixed-point, regularity, and finite-depth certification results. Section~3.2 illustrates the theory through regular and nonregular fixed-point examples and gives concrete witnesses for the incompleteness of the automaton space. Section~4 develops generalisations and extensions of the guarded framework, including higher language powers, special seed languages, finite systems of guarded operators, two-sided Picard approximation, soundness of iteration from below, state-complexity estimates, and depth-capped approximations. Section~5 applies the theoretical results to the construction of a fixed-point-based structural pre-filter for SQL parameter validation and describes the corresponding DFA construction and deployment pipeline. Finally, Section~6 discusses the theoretical and practical limitations of the framework and outlines directions for further research.

\section{Preliminaries}

Throughout the paper, $\Sigma$ denotes a fixed finite nonempty alphabet, and $\Sigma^\ast$ denotes the set of all finite words over $\Sigma$, including the empty word $\varepsilon$.
For a word $w\in\Sigma^\ast$, its length is denoted by $|w|$.
We write $\Lang=\mathcal{P}(\Sigma^\ast)$ for the set of all formal languages over $\Sigma$.
If $L\subseteq\Sigma^\ast$, its characteristic function is denoted by $\one_L$.

A formal language over $\Sigma$ is any subset of $\Sigma^\ast$.
For $u,v\in\Sigma^\ast$ and $L\subseteq\Sigma^\ast$, we use the notation
\[
uLv=\{\,uwv:w\in L\,\}.
\]

\begin{Definition}[{\cite{Salomaa1973,Eilenberg1974}}]
Each language $L\in\Lang$ is identified with its characteristic function
\[
\one_L:\Sigma^\ast\to\{0,1\},\qquad
\one_L(w)=
\begin{cases}
1,&w\in L,\\
0,&w\notin L.
\end{cases}
\]
\end{Definition}

We recall the length-based ultrametric on the space of all formal languages.
\begin{Definition}[{\cite{Vianu1977,FulopKephart2015}}]\label{def:base-ultrametric}
For $L,M\in\Lang$, define
\[
d(L,M)=
\begin{cases}
0, & L=M,\\[2mm]
2^{-n}, & n=\min\{\,|w|:\one_L(w)\neq\one_M(w)\,\}.
\end{cases}
\]
\end{Definition}

If $L\neq M$, then there exists at least one word $w\in\Sigma^\ast$ such that $\one_L(w)\neq\one_M(w)$, so the minimum in Definition~\ref{def:base-ultrametric} is taken over a nonempty subset of $\mathbb{N}\cup\{0\}$;
in particular the formula $2^{-n}$ is well defined. The two-case structure of the definition is exhaustive: the case $L=M$ assigns distance zero, and the case $L\neq M$ assigns $2^{-n}$ where $n$ is the length of the shortest distinguishing word.
Thus $d(L,M)$ is defined for all pairs $(L,M)\in\mathcal{L}\times\mathcal{L}$. The distance $d(L,M)$ is determined by the shortest word on which the two languages differ.
In particular, two languages are close when they agree on all words up to a large length.

\begin{Definition}[{\cite{Schikhof}}]\label{def:ultrametric-space}
Let $X$ be a nonempty set. A function $\rho:X\times X\to[0,\infty)$ is called an ultrametric on $X$ if, for all $x,y,z\in X$, the following conditions hold: $\rho(x,y)=0$ if and only if $x=y$;
$\rho(x,y)=\rho(y,x)$; and $\rho(x,z)\le \max\{\rho(x,y),\rho(y,z)\}$. A pair $(X,\rho)$, where $\rho$ is an ultrametric on $X$, is called an ultrametric space.
\end{Definition}

\begin{Theorem}\cite{Vianu1977,FulopKephart2015}\label{thm:known-language-ultrametric}
The function $d$ is an ultrametric on $\Lang$.
\end{Theorem}

\begin{Theorem}[{\cite{Vianu1977,FulopKephart2015}}]\label{thm:known-language-complete}
The ultrametric space $(\Lang,d)$ is complete.
\end{Theorem}

\begin{Remark}
The completeness of \((\mathcal{L},d)\) can be justified by identifying
\(\mathcal{L}=\mathcal{P}(\Sigma^\ast)\) with \(2^{\Sigma^\ast}\) via
characteristic functions.
The metric/topological structure of language spaces of this type goes back to the Bodnarchuk metric space introduced by Bodnarchuk~\cite{Bodnarchuk1965} and discussed by Vianu~\cite{Vianu1977}.
More generally, complete metric spaces are a standard tool in computer science for solving domain equations in programming semantics~\cite{AmericaRutten1989}.
Fulop and Kephart~\cite{FulopKephart2015} later describe the corresponding Cantor language topology and recall that it is homeomorphic to the Cantor space;
in particular, the language space is compact. Since every compact metric space is complete, this provides a standard justification for the completeness of \((\mathcal{L},d)\).
\end{Remark}

We also recall Banach's fixed-point theorem.

\begin{Theorem}[{\cite{Banach1922}}]\label{thm:banach}
Let $(X,\rho)$ be a complete metric space and let $T:X\to X$ be a contraction;
that is, there exists $q\in(0,1)$ such that
\[
\rho(Tx,Ty)\le q\rho(x,y)
\qquad\text{for all }x,y\in X.
\]
Then $T$ has a unique fixed point $x^\ast\in X$.
Moreover, for every $x_0\in X$, the Picard iteration $x_{n+1}=Tx_n$ converges to $x^\ast$.
\end{Theorem}

For extensions of this fixed-point principle to generalized metric spaces or partially ordered settings, we refer the reader to~\cite{Bhaskar}.

We now recall the deterministic finite automata used in the sequel.

\begin{Definition}\cite{HopcroftMotwaniUllman2007}
A deterministic finite automaton, or DFA, over $\Sigma$ is a tuple
\[
\mathcal{A}=(Q,\Sigma,\delta,q_0,F),
\]
where $Q$ is a finite nonempty set of states, $\delta:Q\times\Sigma\to Q$ is a total transition function, $q_0\in Q$ is the initial state, and $F\subseteq Q$ is the set of accepting states.
The transition function is extended to $\delta:Q\times\Sigma^\ast\to Q$ in the usual way by
\[
\delta(q,\varepsilon)=q,
\qquad
\delta(q,wa)=\delta(\delta(q,w),a)
\]
for all $q\in Q$, $w\in\Sigma^\ast$, and $a\in\Sigma$.
\end{Definition}

\begin{Definition}\cite{HopcroftMotwaniUllman2007}
The language recognized by a DFA $\mathcal{A}=(Q,\Sigma,\delta,q_0,F)$ is
\[
L(\mathcal{A})=\{\,w\in\Sigma^\ast:\delta(q_0,w)\in F\,\}.
\]
A language $L\subseteq\Sigma^\ast$ is called regular if there exists a DFA $\mathcal{A}$ such that $L=L(\mathcal{A})$.
\end{Definition}

We shall use the standard fact that regular languages are closed under finite union and under concatenation with fixed words;
see, for example, \cite{HopcroftMotwaniUllman2007}. In particular, if \(L\) is regular and \(u,v\in\Sigma^\ast\), then \(uLv\) is regular.

Since distinct deterministic finite automata may recognize the same language, we work with automata up to language equivalence.
The standard notion of equivalence of automata is based on equality of the languages they recognize; see, for example, \cite{HopcroftMotwaniUllman2007}.
We now introduce the corresponding quotient space of deterministic finite automata.

\begin{Definition}
Two DFAs $\mathcal{A}$ and $\mathcal{B}$ over $\Sigma$ are called language equivalent if $L(\mathcal{A})=L(\mathcal{B})$.
We denote by $\mathfrak{A}_\Sigma$ the set of equivalence classes of DFAs over $\Sigma$ under language equivalence, and we write $[\mathcal{A}]$ for the equivalence class of an automaton $\mathcal{A}$.
\end{Definition}

For \([\mathcal{A}]\in\mathfrak{A}_\Sigma\), we write \(L([\mathcal{A}]):=L(\mathcal{A})\). This notation is well defined, since all representatives of the same equivalence class recognize the same language.

We finally recall the class of guarded language operators introduced in \cite{HristovIlchevKulinaZlatanov2026}.

\begin{Definition}\cite{HristovIlchevKulinaZlatanov2026} \label{d2.11}
Let $S\subseteq\Sigma^\ast$, let $p\ge 1$, and let
\[
G=\{(u_r,v_r):r=1,\dots,p\}\subseteq\Sigma^\ast\times\Sigma^\ast
\]
be a finite nonempty family of guards.
The corresponding guarded language operator is the map $T_{S,G}:\Lang\to\Lang$ defined by
\[
T_{S,G}(L)=S\cup\bigcup_{r=1}^p u_rLv_r.
\]
The family $G$ is said to have guard length at least $m\ge 1$ if
\[
\min_{1\le r\le p}(|u_r|+|v_r|)\ge m.
\]
\end{Definition}

The following result is the main fixed-point theorem for guarded language operators from \cite{HristovIlchevKulinaZlatanov2026}.
It will be used as the language-level foundation for the automaton results proved in the present paper.

\begin{Theorem}{\cite{HristovIlchevKulinaZlatanov2026}}\label{thm:known-guarded-language}
Let $S\subseteq\Sigma^\ast$, and let $G=\{(u_r,v_r):r=1,\dots,p\}\subseteq\Sigma^\ast\times\Sigma^\ast$ be a finite nonempty family of guards with guard length at least $m\ge 1$.
Then the operator $T_{S,G}:\Lang\to\Lang$ is a contraction on $(\Lang,d)$ with contraction factor $2^{-m}$.
Consequently, $T_{S,G}$ has a unique fixed point $L^\ast\in\Lang$, and for every initial language $L^{(0)}\in\Lang$, the Picard iteration
\[
L^{(n+1)}=T_{S,G}(L^{(n)})
\]
converges to $L^\ast$.
Moreover,
\[
d(L^{(n)},L^\ast)\le 2^{-mn}d(L^{(0)},L^\ast)
\qquad\text{for all }n\ge 0.
\]
\end{Theorem}

\begin{Remark}\label{rem:epsilon-contraction}
A potential concern is whether the empty word $\varepsilon$ could distinguish
$T_{S,G}(L)$ from $T_{S,G}(M)$ and thereby break the contraction estimate.
This cannot occur: $\varepsilon\in T_{S,G}(L)$ if and only if $\varepsilon\in S$
or $\varepsilon\in u_r L v_r$ for some $r$, the latter requiring $u_r=v_r=\varepsilon$
and $\varepsilon\in L$.
Since all guards have length $|u_r|+|v_r|\ge m\ge 1$, the condition $u_r=v_r=\varepsilon$
is excluded.
Therefore $\varepsilon\in T_{S,G}(L)$ if and only if $\varepsilon\in S$,
independently of $L$.
Hence $\varepsilon$ can never distinguish $T_{S,G}(L)$ from
$T_{S,G}(M)$, and the contraction argument is unaffected by the $\varepsilon$ case.
\end{Remark}

%%%%%%%%%%%%%%%%%%%%%%%%%%%%%%%%%%%%

\section{Ultrametric Fixed-Point Theory for Guarded Automata}
\label{sec:main}

\subsection{The Distinguishing-Word Ultrametric and Its Completion}

\begin{Definition}\label{def:automaton-metric}
For $[\mathcal{A}],[\mathcal{B}]\in\mathfrak{A}_\Sigma$, define
\[
d_{\mathrm{aut}}([\mathcal{A}],[\mathcal{B}])=
\begin{cases}
0, & L([\mathcal{A}])=L([\mathcal{B}]),\\[2mm]
2^{-n}, & n=\min\{\,|w|:\one_{L([\mathcal{A}])}(w)\neq \one_{L([\mathcal{B}])}(w)\,\}.
\end{cases}
\]
\end{Definition}
Note that if the shortest distinguishing word is the empty word $\varepsilon$, then $n=0$ and the distance evaluates to $2^{-0} = 1$.

The quantity $d_{\mathrm{aut}}([\mathcal{A}],[\mathcal{B}])$ measures the length of the shortest input word on which the two automata differ in acceptance.
It depends only on the recognized languages and is therefore well defined on equivalence classes.

\begin{Lemma}\label{lem:automaton-ultrametric}
The function $d_{\mathrm{aut}}$ is an ultrametric on $\mathfrak{A}_\Sigma$.
\end{Lemma}

\begin{proof}
We first show that $d_{\mathrm{aut}}$ is well defined.
Let $[\mathcal{A}],[\mathcal{B}]\in\mathfrak{A}_\Sigma$, and suppose that $\mathcal{A}'\in[\mathcal{A}]$ and $\mathcal{B}'\in[\mathcal{B}]$ are different representatives of the same equivalence classes.
By the definition of language equivalence, we have $L(\mathcal{A}')=L(\mathcal{A})$ and $L(\mathcal{B}')=L(\mathcal{B})$.
Hence the set of words on which $\mathcal{A}$ and $\mathcal{B}$ differ in acceptance is the same as the set of words on which $\mathcal{A}'$ and $\mathcal{B}'$ differ in acceptance.
Therefore the value of $d_{\mathrm{aut}}([\mathcal{A}],[\mathcal{B}])$ does not depend on the chosen representatives.

Non-negativity is immediate, since by definition $d_{\mathrm{aut}}$ takes only the values $0$ and $2^{-n}$ for integers $n\ge 0$.
Symmetry is also immediate, because the condition $\one_{L([\mathcal{A}])}(w)\neq \one_{L([\mathcal{B}])}(w)$ is symmetric in $[\mathcal{A}]$ and $[\mathcal{B}]$.
We next verify the identity of indiscernibles. If $[\mathcal{A}]=[\mathcal{B}]$, then $L([\mathcal{A}])=L([\mathcal{B}])$, and therefore $d_{\mathrm{aut}}([\mathcal{A}],[\mathcal{B}])=0$.
Conversely, if $d_{\mathrm{aut}}([\mathcal{A}],[\mathcal{B}])=0$, then by the definition of $d_{\mathrm{aut}}$ we have $L([\mathcal{A}])=L([\mathcal{B}])$.
Hence $\mathcal{A}$ and $\mathcal{B}$ are language equivalent, and therefore $[\mathcal{A}]=[\mathcal{B}]$.

It remains to prove the strong triangle inequality.
Define $\iota:\mathfrak{A}_\Sigma\to\Lang$ by $\iota([\mathcal{A}])=L([\mathcal{A}])$. This map is well defined because all representatives of the same automaton class recognize the same language.
By the definition of $d_{\mathrm{aut}}$ and by the definition of the language ultrametric $d$, for all $[\mathcal{A}],[\mathcal{B}]\in\mathfrak{A}_\Sigma$ we have $d_{\mathrm{aut}}([\mathcal{A}],[\mathcal{B}])=d(\iota([\mathcal{A}]),\iota([\mathcal{B}]))$.
Now let $[\mathcal{A}],[\mathcal{B}],[\mathcal{C}]\in\mathfrak{A}_\Sigma$ be arbitrary. Since $d$ is an ultrametric on $\Lang$ by Theorem~\ref{thm:known-language-ultrametric}, we have $d(\iota([\mathcal{A}]),\iota([\mathcal{C}]))\le \max\{d(\iota([\mathcal{A}]),\iota([\mathcal{B}])),d(\iota([\mathcal{B}]),\iota([\mathcal{C}]))\}$.
Using the equality between $d_{\mathrm{aut}}$ and the pullback of $d$ under $\iota$, this becomes $d_{\mathrm{aut}}([\mathcal{A}],[\mathcal{C}])\le \max\{d_{\mathrm{aut}}([\mathcal{A}],[\mathcal{B}]),d_{\mathrm{aut}}([\mathcal{B}],[\mathcal{C}])\}$.
Thus $d_{\mathrm{aut}}$ satisfies non-negativity, symmetry, the identity of indiscernibles, and the strong triangle inequality.
Therefore $d_{\mathrm{aut}}$ is an ultrametric on $\mathfrak{A}_\Sigma$.
\end{proof}

\begin{Proposition}\label{prop:isometric-embedding}
The map $\iota:\mathfrak{A}_\Sigma\to\Lang$, defined by $\iota([\mathcal{A}])=L([\mathcal{A}])$, is an isometric embedding of $\mathfrak{A}_\Sigma$ onto the subspace $\Lang_{\mathrm{reg}}:=\{\,L\subseteq\Sigma^\ast : L \text{ is regular}\,\}$.
\end{Proposition}

\begin{proof}
We first prove that the map $\iota$ is well defined. Let $[\mathcal{A}]\in\mathfrak{A}_\Sigma$.
By definition, $[\mathcal{A}]$ is an equivalence class of deterministic finite automata under language equivalence.
Suppose that $\mathcal{A}'$ is another representative of the same class, that is, $[\mathcal{A}']=[\mathcal{A}]$.
Then $\mathcal{A}'$ and $\mathcal{A}$ are language equivalent, so $L(\mathcal{A}')=L(\mathcal{A})$. Therefore the value assigned by $\iota$ does not depend on the chosen representative of the class.
Hence $\iota([\mathcal{A}])=L([\mathcal{A}])$ is well defined.

Next we prove that $\iota$ is injective. Let $[\mathcal{A}],[\mathcal{B}]\in\mathfrak{A}_\Sigma$ and assume that $\iota([\mathcal{A}])=\iota([\mathcal{B}])$.
By the definition of $\iota$, this means that $L([\mathcal{A}])=L([\mathcal{B}])$. Equivalently, $L(\mathcal{A})=L(\mathcal{B})$. Hence $\mathcal{A}$ and $\mathcal{B}$ are language equivalent.
By the definition of the equivalence relation on deterministic finite automata, this implies $[\mathcal{A}]=[\mathcal{B}]$. Thus $\iota$ is injective.

We now prove that the image of $\iota$ is exactly $\Lang_{\mathrm{reg}}$. First, let $[\mathcal{A}]\in\mathfrak{A}_\Sigma$.
Since $\mathcal{A}$ is a DFA, the language $L(\mathcal{A})$ recognized by $\mathcal{A}$ is regular. Hence $\iota([\mathcal{A}])=L([\mathcal{A}])=L(\mathcal{A})\in\Lang_{\mathrm{reg}}$. Therefore $\iota(\mathfrak{A}_\Sigma)\subseteq\Lang_{\mathrm{reg}}$.

Conversely, let $L\in\Lang_{\mathrm{reg}}$.
By the definition of a regular language, there exists a DFA$\mathcal{A}$ over $\Sigma$ such that $L(\mathcal{A})=L$.
Then $[\mathcal{A}]\in\mathfrak{A}_\Sigma$, and by the definition of $\iota$ we have $\iota([\mathcal{A}])=L([\mathcal{A}])=L(\mathcal{A})=L$.
Hence every regular language belongs to the image of $\iota$, so $\Lang_{\mathrm{reg}}\subseteq\iota(\mathfrak{A}_\Sigma)$. Combining the two inclusions, we obtain $\iota(\mathfrak{A}_\Sigma)=\Lang_{\mathrm{reg}}$.

It remains to prove that $\iota$ preserves distances. Let $[\mathcal{A}],[\mathcal{B}]\in\mathfrak{A}_\Sigma$ be arbitrary. We consider two cases.

First, suppose that $[\mathcal{A}]=[\mathcal{B}]$.
Then $L([\mathcal{A}])=L([\mathcal{B}])$, and by Definition~\ref{def:automaton-metric} we have $d_{\mathrm{aut}}([\mathcal{A}],[\mathcal{B}])=0$. On the other hand, $\iota([\mathcal{A}])=\iota([\mathcal{B}])$, so $d(\iota([\mathcal{A}]),\iota([\mathcal{B}]))=0$. Hence in this case $d_{\mathrm{aut}}([\mathcal{A}],[\mathcal{B}])=d(\iota([\mathcal{A}]),\iota([\mathcal{B}]))$.

Second, suppose that $[\mathcal{A}]\neq[\mathcal{B}]$. Then $L([\mathcal{A}])\neq L([\mathcal{B}])$. By the definition of $d_{\mathrm{aut}}$, the distance $d_{\mathrm{aut}}([\mathcal{A}],[\mathcal{B}])$ is equal to $2^{-n}$, where $n$ is the minimum length of a word $w\in\Sigma^\ast$ such that $\one_{L([\mathcal{A}])}(w)\neq\one_{L([\mathcal{B}])}(w)$.
Since $\iota([\mathcal{A}])=L([\mathcal{A}])$ and $\iota([\mathcal{B}])=L([\mathcal{B}])$, the same word $w$ is precisely a distinguishing word between the languages $\iota([\mathcal{A}])$ and $\iota([\mathcal{B}])$.
Therefore, by the definition of the language ultrametric $d$, we also have $d(\iota([\mathcal{A}]),\iota([\mathcal{B}]))=2^{-n}$. Hence again $d_{\mathrm{aut}}([\mathcal{A}],[\mathcal{B}])=d(\iota([\mathcal{A}]),\iota([\mathcal{B}]))$.

Thus, for all $[\mathcal{A}],[\mathcal{B}]\in\mathfrak{A}_\Sigma$, we have $d_{\mathrm{aut}}([\mathcal{A}],[\mathcal{B}])=d(\iota([\mathcal{A}]),\iota([\mathcal{B}]))$. Therefore $\iota$ is distance preserving.
Since we have already proved that $\iota$ is injective and that its image is exactly $\Lang_{\mathrm{reg}}$, it follows that $\iota$ is an isometric embedding of $\mathfrak{A}_\Sigma$ onto $\Lang_{\mathrm{reg}}$.
\end{proof}

The automaton space is naturally identified with the regular languages.
The next result shows that this space is not complete with respect to the distinguishing-word ultrametric.

\begin{Proposition}\label{prop:not-complete}
The ultrametric space $(\mathfrak{A}_\Sigma,d_{\mathrm{aut}})$ is not complete.
\end{Proposition}

\begin{proof}
Since $\Sigma$ is nonempty, choose a letter $a\in\Sigma$. Consider the language $K=\{a^{n^2}:n\ge 0\}$.
We first show that $K$ is not regular. Indeed, suppose for contradiction that $K$ is regular, and let $P$ be its pumping length.
Choose an integer $r>P$ and consider the word $a^{r^2}\in K$.
By the pumping lemma, there exist words $x,y,z\in\Sigma^\ast$ such that $a^{r^2}=xyz$, $|xy|\le P$, $|y|\ge 1$, and $xy^iz\in K$ for all $i\ge 0$.
Since $a^{r^2}$ is a unary word, we have $y=a^t$ for some integer $t$ with $1\le t\le P$.
Pumping once gives $xy^2z=a^{r^2+t}$. Because $1\le t\le P<r$, we have $r^2<r^2+t<r^2+r<(r+1)^2$. Hence $r^2+t$ is not a square, so $a^{r^2+t}\notin K$.
This contradicts the pumping lemma. Therefore $K$ is not regular.

For each integer $N\ge 1$, define $K_N:=K\cap\{\,w\in\Sigma^\ast:|w|<N\,\}$.
Each language $K_N$ is finite, because there are only finitely many words of length strictly less than $N$ over the finite alphabet $\Sigma$.
Hence each $K_N$ is regular. Therefore, for every $N\ge 1$, there exists a DFA $\mathcal{A}_N$ such that $L(\mathcal{A}_N)=K_N$.
Let $[\mathcal{A}_N]\in\mathfrak{A}_\Sigma$ be its equivalence class.

We prove that the sequence $([\mathcal{A}_N])_{N\ge 1}$ is Cauchy in $(\mathfrak{A}_\Sigma,d_{\mathrm{aut}})$.
Let $\varepsilon>0$ be arbitrary. Choose an integer $N_0$ such that $2^{-N_0}<\varepsilon$. Now let $M,N\ge N_0$.
By the definition of $K_M$ and $K_N$, both languages agree with $K$ on all words of length strictly less than $N_0$.
Therefore $K_M$ and $K_N$ agree on all words of length strictly less than $N_0$.
Hence the first possible disagreement between $K_M$ and $K_N$, if any, can occur only at length at least $N_0$.
By the definition of $d_{\mathrm{aut}}$, this gives $d_{\mathrm{aut}}([\mathcal{A}_M],[\mathcal{A}_N])\le 2^{-N_0}<\varepsilon$. Thus $([\mathcal{A}_N])_{N\ge 1}$ is a Cauchy sequence in $(\mathfrak{A}_\Sigma,d_{\mathrm{aut}})$.

It remains to show that this Cauchy sequence has no limit in $\mathfrak{A}_\Sigma$.
Suppose, for contradiction, that there exists $[\mathcal{A}]\in\mathfrak{A}_\Sigma$ such that $[\mathcal{A}_N]\to[\mathcal{A}]$ in $d_{\mathrm{aut}}$.
By Proposition~\ref{prop:isometric-embedding}, the map $\iota([\mathcal{B}])=L([\mathcal{B}])$ is an isometry from $\mathfrak{A}_\Sigma$ onto $\Lang_{\mathrm{reg}}$.
Therefore the convergence $[\mathcal{A}_N]\to[\mathcal{A}]$ implies $K_N=L([\mathcal{A}_N])\to L([\mathcal{A}])$ in $(\Lang,d)$.

On the other hand, the sequence $(K_N)$ converges to $K$ in $(\Lang,d)$.
Indeed, for every $N\ge 1$, the languages $K_N$ and $K$ agree on all words of length strictly less than $N$, because $K_N$ is exactly the restriction of $K$ to words of length less than $N$.
Hence $d(K_N,K)\le 2^{-N}$, and therefore $d(K_N,K)\to 0$ as $N\to\infty$. Thus $(K_N)$ has the limit $K$ in $(\Lang,d)$.
Since limits in a metric space are unique, we obtain $L([\mathcal{A}])=K$.
But $L([\mathcal{A}])$ is recognized by a DFA and is therefore regular, while $K$ is not regular.
This contradiction shows that the Cauchy sequence $([\mathcal{A}_N])_{N\ge 1}$ does not converge in $\mathfrak{A}_\Sigma$. Consequently, $(\mathfrak{A}_\Sigma,d_{\mathrm{aut}})$ is not complete.
\end{proof}

\begin{Proposition}\label{prop:dense}
The subspace $\Lang_{\mathrm{reg}}$ is dense in $(\Lang,d)$. Consequently, the completion of $(\mathfrak{A}_\Sigma,d_{\mathrm{aut}})$ can be naturally identified with $(\Lang,d)$.
\end{Proposition}

\begin{proof}
We first prove that $\Lang_{\mathrm{reg}}$ is dense in $(\Lang,d)$. Let $L\subseteq\Sigma^\ast$ be arbitrary.
We have to show that, for every accuracy level, there exists a regular language which is sufficiently close to $L$ in the metric $d$.
Let $N\ge 1$ be fixed, and define $L^{<N}:=L\cap\{\,w\in\Sigma^\ast:|w|<N\,\}$. Since $\Sigma$ is finite, there are only finitely many words over $\Sigma$ of length strictly less than $N$.
Hence the set $\{\,w\in\Sigma^\ast:|w|<N\,\}$ is finite, and therefore $L^{<N}$ is also finite. Every finite language is regular, so $L^{<N}\in\Lang_{\mathrm{reg}}$.
By construction, $L^{<N}$ contains exactly those words of $L$ whose length is strictly less than $N$.
Hence, for every word $w\in\Sigma^\ast$ with $|w|<N$, we have $w\in L^{<N}$ if and only if $w\in L$.
Equivalently, $\one_{L^{<N}}(w)=\one_L(w)$ for all $w\in\Sigma^\ast$ with $|w|<N$. Therefore the first possible disagreement between $L^{<N}$ and $L$, if such a disagreement exists, can occur only at a word of length at least $N$.
By the definition of the language ultrametric, this implies $d(L^{<N},L)\le 2^{-N}$.
Since $N$ can be chosen arbitrarily large, the regular languages $L^{<N}$ approximate $L$ arbitrarily well in $(\Lang,d)$.
More precisely, for every $\varepsilon>0$, we may choose $N$ so large that $2^{-N}<\varepsilon$, and then the regular language $L^{<N}$ satisfies $d(L^{<N},L)<\varepsilon$.
Thus every language $L\in\Lang$ lies in the closure of $\Lang_{\mathrm{reg}}$, and consequently $\Lang_{\mathrm{reg}}$ is dense in $(\Lang,d)$.

We now identify the completion of the automaton space. By Proposition~\ref{prop:isometric-embedding}, the map $\iota:\mathfrak{A}_\Sigma\to\Lang$ given by $\iota([\mathcal{A}])=L([\mathcal{A}])$ is an isometry from $(\mathfrak{A}_\Sigma,d_{\mathrm{aut}})$ onto $\Lang_{\mathrm{reg}}$ equipped with the restriction of the metric $d$.
Therefore $(\mathfrak{A}_\Sigma,d_{\mathrm{aut}})$ and $(\Lang_{\mathrm{reg}},d)$ have the same metric completion up to canonical isometry.

By Theorem~\ref{thm:known-language-complete}, the space $(\Lang,d)$ is complete.
Since we have just proved that $\Lang_{\mathrm{reg}}$ is dense in $(\Lang,d)$, the complete space $(\Lang,d)$ is precisely the completion of $(\Lang_{\mathrm{reg}},d)$.
Using the isometry between $(\mathfrak{A}_\Sigma,d_{\mathrm{aut}})$ and $(\Lang_{\mathrm{reg}},d)$, it follows that the completion of $(\mathfrak{A}_\Sigma,d_{\mathrm{aut}})$ can be naturally identified with $(\Lang,d)$.
\end{proof}

\subsection{Guarded Automaton Contractions and Fixed Points}

We now pass from the language-level guarded operators recalled in Theorem~\ref{thm:known-guarded-language} to their induced action on automata classes.

\begin{Definition}\label{def:automaton-operator}
Let $S\subseteq\Sigma^\ast$ be a regular language and let $G=\{(u_r,v_r):r=1,\dots,p\}\subseteq\Sigma^\ast\times\Sigma^\ast$ be a finite family of guards.
The guarded language operator $T_{S,G}$ maps regular languages to regular languages, because regular languages are closed under finite union and concatenation with fixed words.
Hence it induces an operator $\mathbb{T}_{S,G}:\mathfrak{A}_\Sigma\to\mathfrak{A}_\Sigma$ by $\mathbb{T}_{S,G}([\mathcal{A}])=[\mathcal{B}]$, where $\mathcal{B}$ is any DFA recognizing $T_{S,G}(L([\mathcal{A}]))$.
\end{Definition}

\begin{Remark}
The operator $\mathbb{T}_{S,G}$ is well defined. Indeed, if $\mathcal{A}'$ is another representative of $[\mathcal{A}]$, then $L(\mathcal{A}')=L(\mathcal{A})$, and hence $T_{S,G}(L(\mathcal{A}'))=T_{S,G}(L(\mathcal{A}))$.
Moreover, if $\mathcal{B}$ and $\mathcal{B}'$ are two deterministic finite automata recognizing $T_{S,G}(L([\mathcal{A}]))$, then $L(\mathcal{B})=L(\mathcal{B}')$, so $[\mathcal{B}]=[\mathcal{B}']$.
Therefore $\mathbb{T}_{S,G}([\mathcal{A}])$ does not depend on any of the choices involved.
\end{Remark}

\begin{Theorem}\label{thm:automaton-contraction}
Let $S\subseteq\Sigma^\ast$ be regular, and let $G=\{(u_r,v_r):r=1,\dots,p\}$ be a finite family of guards with guard length at least $m\ge 1$.
Then the induced automaton operator $\mathbb{T}_{S,G}$ is a contraction on $(\mathfrak{A}_\Sigma,d_{\mathrm{aut}})$ with contraction factor $2^{-m}$;
that is, $d_{\mathrm{aut}}(\mathbb{T}_{S,G}([\mathcal{A}]),\mathbb{T}_{S,G}([\mathcal{B}]))\le 2^{-m}d_{\mathrm{aut}}([\mathcal{A}],[\mathcal{B}])$ for all $[\mathcal{A}],[\mathcal{B}]\in\mathfrak{A}_\Sigma$.
\end{Theorem}

\begin{proof}
Let $[\mathcal{A}],[\mathcal{B}]\in\mathfrak{A}_\Sigma$ be arbitrary.
We have to prove that applying the induced operator $\mathbb{T}_{S,G}$ decreases the automaton distance by at least the factor $2^{-m}$.
By Definition~\ref{def:automaton-operator}, the class $\mathbb{T}_{S,G}([\mathcal{A}])$ is represented by any DFA recognizing the language $T_{S,G}(L([\mathcal{A}]))$.
Therefore the language recognized by the class $\mathbb{T}_{S,G}([\mathcal{A}])$ is exactly $T_{S,G}(L([\mathcal{A}]))$.
In the same way, the language recognized by the class $\mathbb{T}_{S,G}([\mathcal{B}])$ is exactly $T_{S,G}(L([\mathcal{B}]))$. Thus we have $L(\mathbb{T}_{S,G}([\mathcal{A}]))=T_{S,G}(L([\mathcal{A}]))$ and $L(\mathbb{T}_{S,G}([\mathcal{B}]))=T_{S,G}(L([\mathcal{B}]))$.
Now, by the definition of the automaton metric $d_{\mathrm{aut}}$, the distance between two automaton classes is the language ultrametric distance between the languages recognized by these classes.
Hence
\[
d_{\mathrm{aut}}(\mathbb{T}_{S,G}([\mathcal{A}]),\mathbb{T}_{S,G}([\mathcal{B}]))=d(L(\mathbb{T}_{S,G}([\mathcal{A}])),L(\mathbb{T}_{S,G}([\mathcal{B}]))).
\]
Using the two identities above, this becomes
\[
d_{\mathrm{aut}}(\mathbb{T}_{S,G}([\mathcal{A}]),\mathbb{T}_{S,G}([\mathcal{B}]))=d(T_{S,G}(L([\mathcal{A}])),T_{S,G}(L([\mathcal{B}]))).
\]

Since the family of guards $G$ has guard length at least $m$, Theorem~\ref{thm:known-guarded-language} applies to the language operator $T_{S,G}$.
Therefore $T_{S,G}$ is a contraction on $(\Lang,d)$ with contraction factor $2^{-m}$.
Applying this result to the two languages $L([\mathcal{A}])$ and $L([\mathcal{B}])$, we obtain
\[
d(T_{S,G}(L([\mathcal{A}])),T_{S,G}(L([\mathcal{B}])))\le 2^{-m}d(L([\mathcal{A}]),L([\mathcal{B}])).
\]

Finally, by Proposition~\ref{prop:isometric-embedding}, the map $\iota:\mathfrak{A}_\Sigma\to\Lang$ given by $\iota([\mathcal{C}])=L([\mathcal{C}])$ is an isometry onto $\Lang_{\mathrm{reg}}$.
Hence, for the two automaton classes $[\mathcal{A}]$ and $[\mathcal{B}]$, we have $d(L([\mathcal{A}]),L([\mathcal{B}]))=d_{\mathrm{aut}}([\mathcal{A}],[\mathcal{B}])$.

Combining the previous estimates gives $d_{\mathrm{aut}}(\mathbb{T}_{S,G}([\mathcal{A}]),\mathbb{T}_{S,G}([\mathcal{B}]))\le 2^{-m}d_{\mathrm{aut}}([\mathcal{A}],[\mathcal{B}])$.
Since $[\mathcal{A}]$ and $[\mathcal{B}]$ were arbitrary, the induced operator $\mathbb{T}_{S,G}$ is a contraction on $(\mathfrak{A}_\Sigma,d_{\mathrm{aut}})$ with contraction factor $2^{-m}$.
\end{proof}

The previous theorem shows that the induced automaton operator is contractive.
However, Proposition~\ref{prop:not-complete} shows that the automaton space is not complete.
Thus Banach's fixed-point theorem cannot in general be applied inside $\mathfrak{A}_\Sigma$.
The correct fixed-point object lies in the completion, which is naturally identified with the full language space $(\Lang,d)$.
Thus the automaton iteration does not necessarily converge to an automaton, but it always determines a convergent sequence after passing to the completion.

\begin{Theorem}\label{thm:automaton-cauchy}
Under the assumptions of Theorem~\ref{thm:automaton-contraction}, let $[\mathcal{A}^{(0)}]\in\mathfrak{A}_\Sigma$ be arbitrary and define the Picard iteration by $[\mathcal{A}^{(n+1)}]=\mathbb{T}_{S,G}([\mathcal{A}^{(n)}])$ for $n\ge 0$.
Then the sequence $([\mathcal{A}^{(n)}])_{n\ge 0}$ is Cauchy in $(\mathfrak{A}_\Sigma,d_{\mathrm{aut}})$. More precisely, for all integers $q>n\ge 0$, one has $d_{\mathrm{aut}}([\mathcal{A}^{(q)}],[\mathcal{A}^{(n)}])\le 2^{-mn}d_{\mathrm{aut}}([\mathcal{A}^{(1)}],[\mathcal{A}^{(0)}])$.
Moreover, if $L_n:=L([\mathcal{A}^{(n)}])$, then $L_{n+1}=T_{S,G}(L_n)$ for all $n\ge 0$, and the sequence $(L_n)$ converges in $(\Lang,d)$ to the unique language fixed point $L^\ast$ of $T_{S,G}$.
\end{Theorem}

\begin{proof}
Let $[\mathcal{A}^{(0)}]\in\mathfrak{A}_\Sigma$ be fixed, and let $([\mathcal{A}^{(n)}])_{n\ge 0}$ be the Picard sequence generated by the induced automaton operator $\mathbb{T}_{S,G}$.
By Theorem~\ref{thm:automaton-contraction}, the operator $\mathbb{T}_{S,G}$ is a contraction on $(\mathfrak{A}_\Sigma,d_{\mathrm{aut}})$ with contraction factor $2^{-m}$.
Therefore, for every $j\ge 1$, we have $d_{\mathrm{aut}}([\mathcal{A}^{(j+1)}],[\mathcal{A}^{(j)}])\le 2^{-m}d_{\mathrm{aut}}([\mathcal{A}^{(j)}],[\mathcal{A}^{(j-1)}])$.

We first prove by induction that, for every $j\ge 0$, the distance between two consecutive iterates satisfies $d_{\mathrm{aut}}([\mathcal{A}^{(j+1)}],[\mathcal{A}^{(j)}])\le 2^{-mj}d_{\mathrm{aut}}([\mathcal{A}^{(1)}],[\mathcal{A}^{(0)}])$.
For $j=0$, this is just the identity $d_{\mathrm{aut}}([\mathcal{A}^{(1)}],[\mathcal{A}^{(0)}])\le d_{\mathrm{aut}}([\mathcal{A}^{(1)}],[\mathcal{A}^{(0)}])$. Suppose now that the estimate holds for some $j\ge 0$.
Applying the contraction property to the pair $[\mathcal{A}^{(j+1)}]$ and $[\mathcal{A}^{(j)}]$, we get
\[
d_{\mathrm{aut}}([\mathcal{A}^{(j+2)}],[\mathcal{A}^{(j+1)}])\le 2^{-m}d_{\mathrm{aut}}([\mathcal{A}^{(j+1)}],[\mathcal{A}^{(j)}])\le 2^{-m}2^{-mj}d_{\mathrm{aut}}([\mathcal{A}^{(1)}],[\mathcal{A}^{(0)}]).
\]
Thus $d_{\mathrm{aut}}([\mathcal{A}^{(j+2)}],[\mathcal{A}^{(j+1)}])\le 2^{-m(j+1)}d_{\mathrm{aut}}([\mathcal{A}^{(1)}],[\mathcal{A}^{(0)}])$, which completes the induction.

We now prove the announced estimate for two arbitrary iterates. Let $q>n\ge 0$.
Since $d_{\mathrm{aut}}$ is an ultrametric by Lemma~\ref{lem:automaton-ultrametric}, repeated use of the strong triangle inequality gives
\[
d_{\mathrm{aut}}([\mathcal{A}^{(q)}],[\mathcal{A}^{(n)}])\le \max_{n\le j\le q-1}d_{\mathrm{aut}}([\mathcal{A}^{(j+1)}],[\mathcal{A}^{(j)}]).
\]
By the estimate for consecutive iterates proved above, for every $j$ with $n\le j\le q-1$, we have $d_{\mathrm{aut}}([\mathcal{A}^{(j+1)}],[\mathcal{A}^{(j)}])\le 2^{-mj}d_{\mathrm{aut}}([\mathcal{A}^{(1)}],[\mathcal{A}^{(0)}])$.
Since $j\ge n$ and $m\ge 1$, we have $2^{-mj}\le 2^{-mn}$.
Hence each term in the maximum is bounded above by $2^{-mn}d_{\mathrm{aut}}([\mathcal{A}^{(1)}],[\mathcal{A}^{(0)}])$. Therefore
\[
d_{\mathrm{aut}}([\mathcal{A}^{(q)}],[\mathcal{A}^{(n)}])\le 2^{-mn}d_{\mathrm{aut}}([\mathcal{A}^{(1)}],[\mathcal{A}^{(0)}]).
\]
This proves the quantitative estimate.

We now show that the sequence is Cauchy. Let $\varepsilon>0$ be arbitrary.
Since $m\ge 1$, we have $2^{-mn}\to 0$ as $n\to\infty$. Hence there exists an integer $N$ such that $2^{-mN}d_{\mathrm{aut}}([\mathcal{A}^{(1)}],[\mathcal{A}^{(0)}])<\varepsilon$.
If $q>n\ge N$, then the estimate above gives $d_{\mathrm{aut}}([\mathcal{A}^{(q)}],[\mathcal{A}^{(n)}])<\varepsilon$. If $q=n$, the distance is $0$, and if $q<n$, the same conclusion follows by symmetry of $d_{\mathrm{aut}}$.
Thus $([\mathcal{A}^{(n)}])_{n\ge 0}$ is a Cauchy sequence in $(\mathfrak{A}_\Sigma,d_{\mathrm{aut}})$.

It remains to identify the limit in the completion.
Define $L_n:=L([\mathcal{A}^{(n)}])$ for every $n\ge 0$. By the definition of the induced operator $\mathbb{T}_{S,G}$, the class $[\mathcal{A}^{(n+1)}]=\mathbb{T}_{S,G}([\mathcal{A}^{(n)}])$ recognizes the language $T_{S,G}(L([\mathcal{A}^{(n)}]))$.
Therefore $L_{n+1}=T_{S,G}(L_n)$ for every $n\ge 0$. Hence the sequence $(L_n)$ is exactly the Picard iteration of the guarded language operator $T_{S,G}$ starting from $L_0=L([\mathcal{A}^{(0)}])$.
By Theorem~\ref{thm:known-guarded-language}, the language operator $T_{S,G}$ has a unique fixed point $L^\ast\in\Lang$, and its Picard iteration from any initial language converges to $L^\ast$ in $(\Lang,d)$.
Applying this to the initial language $L_0$, we obtain $L_n\to L^\ast$ in $(\Lang,d)$.
Thus the automaton Picard sequence is Cauchy in the automaton space, and its image under the natural embedding into the completion converges to the unique language fixed point $L^\ast$.
\end{proof}

\begin{Theorem}\label{thm:automaton-fixed-point-iff-regular}
Under the assumptions of Theorem~\ref{thm:automaton-contraction}, the following are equivalent:
\begin{enumerate}[label=\textup{(\roman*)}]
\item the induced automaton operator $\mathbb{T}_{S,G}$ has a fixed point in $\mathfrak{A}_\Sigma$;
\item the unique fixed point $L^\ast$ of the language operator $T_{S,G}$ is regular.
\end{enumerate}
In this case the fixed point in the automaton quotient space is unique.
\end{Theorem}

\begin{proof}
We prove the equivalence in two directions.
First, assume that the induced automaton operator $\mathbb{T}_{S,G}$ has a fixed point in $\mathfrak{A}_\Sigma$.
Thus there exists an automaton class $[\mathcal{A}^\ast]\in\mathfrak{A}_\Sigma$ such that $\mathbb{T}_{S,G}([\mathcal{A}^\ast])=[\mathcal{A}^\ast]$.
By the definition of the induced operator, the class $\mathbb{T}_{S,G}([\mathcal{A}^\ast])$ recognizes the language $T_{S,G}(L([\mathcal{A}^\ast]))$.
Since this class is equal to $[\mathcal{A}^\ast]$, the two classes recognize the same language. Therefore $T_{S,G}(L([\mathcal{A}^\ast]))=L([\mathcal{A}^\ast])$.
This means that the language $L([\mathcal{A}^\ast])$ is a fixed point of the language operator $T_{S,G}$ in the full language space $(\Lang,d)$.
By Theorem~\ref{thm:known-guarded-language}, the operator $T_{S,G}$ has a unique fixed point in $\Lang$, denoted by $L^\ast$. Hence $L([\mathcal{A}^\ast])=L^\ast$.
But $L([\mathcal{A}^\ast])$ is recognized by a DFA, namely any representative of the class $[\mathcal{A}^\ast]$. Therefore $L([\mathcal{A}^\ast])$ is regular.
Since $L([\mathcal{A}^\ast])=L^\ast$, it follows that $L^\ast$ is regular.

Conversely, assume that the unique language fixed point $L^\ast$ of $T_{S,G}$ is regular.
By the definition of regularity, there exists a DFA $\mathcal{A}^\ast$ such that $L(\mathcal{A}^\ast)=L^\ast$. Equivalently, in terms of equivalence classes, $L([\mathcal{A}^\ast])=L^\ast$.
Since $L^\ast$ is a fixed point of $T_{S,G}$, we have $T_{S,G}(L^\ast)=L^\ast$. Substituting $L([\mathcal{A}^\ast])=L^\ast$, we get $T_{S,G}(L([\mathcal{A}^\ast]))=L([\mathcal{A}^\ast])$.
By the definition of the induced automaton operator, the class $\mathbb{T}_{S,G}([\mathcal{A}^\ast])$ is represented by a DFA recognizing the language $T_{S,G}(L([\mathcal{A}^\ast]))$.
But we have just shown that $T_{S,G}(L([\mathcal{A}^\ast]))=L([\mathcal{A}^\ast])$. Therefore $\mathbb{T}_{S,G}([\mathcal{A}^\ast])$ and $[\mathcal{A}^\ast]$ recognize the same language.
Hence they are the same equivalence class, that is, $\mathbb{T}_{S,G}([\mathcal{A}^\ast])=[\mathcal{A}^\ast]$.
Thus $[\mathcal{A}^\ast]$ is a fixed point of the induced automaton operator.

It remains to prove uniqueness of the fixed point in the automaton quotient space.
Suppose that $[\mathcal{A}^\ast]$ and $[\mathcal{B}^\ast]$ are two fixed points of $\mathbb{T}_{S,G}$. Then $\mathbb{T}_{S,G}([\mathcal{A}^\ast])=[\mathcal{A}^\ast]$ and $\mathbb{T}_{S,G}([\mathcal{B}^\ast])=[\mathcal{B}^\ast]$.
By the same argument as above, this implies that $L([\mathcal{A}^\ast])$ and $L([\mathcal{B}^\ast])$ are both fixed points of the language operator $T_{S,G}$.
Since the language fixed point is unique by Theorem~\ref{thm:known-guarded-language}, we obtain $L([\mathcal{A}^\ast])=L^\ast=L([\mathcal{B}^\ast])$.
Hence the automata $\mathcal{A}^\ast$ and $\mathcal{B}^\ast$ are language equivalent, and therefore $[\mathcal{A}^\ast]=[\mathcal{B}^\ast]$.
Thus the fixed point in the automaton quotient space is unique.
\end{proof}

\begin{Corollary}\label{cor:automaton-certified-depth}
Let $[\mathcal{A}^{(n+1)}]=\mathbb{T}_{S,G}([\mathcal{A}^{(n)}])$ for $n\ge 0$ be the automaton Picard iteration, and let $L^\ast$ be the unique language fixed point of $T_{S,G}$.
If $2^{-mn}d(L([\mathcal{A}^{(0)}]),L^\ast)\le 2^{-N}$ for some integer $N\ge 1$, then $\one_{L([\mathcal{A}^{(n)}])}(w)=\one_{L^\ast}(w)$ for all $w\in\Sigma^\ast$ with $|w|<N$.
In particular, if the initial automaton recognizes the empty language, then $n\ge \left\lceil N/m\right\rceil$ is sufficient.
\end{Corollary}

\begin{proof}
Let $L_n:=L([\mathcal{A}^{(n)}])$ for every $n\ge 0$. By the definition of the induced automaton operator, the sequence $(L_n)$ is exactly the Picard iteration of the language operator $T_{S,G}$ starting from $L_0=L([\mathcal{A}^{(0)}])$.
Therefore, by Theorem~\ref{thm:known-guarded-language}, we have $d(L_n,L^\ast)\le 2^{-mn}d(L_0,L^\ast)$ for every $n\ge 0$.
Rewriting this in terms of automata, we obtain $d(L([\mathcal{A}^{(n)}]),L^\ast)\le 2^{-mn}d(L([\mathcal{A}^{(0)}]),L^\ast)$.

Assume now that $2^{-mn}d(L([\mathcal{A}^{(0)}]),L^\ast)\le 2^{-N}$.
Combining this assumption with the previous estimate gives $d(L([\mathcal{A}^{(n)}]),L^\ast)\le 2^{-N}$.
We prove that this implies agreement on all words of length strictly less than $N$.
If $d(L([\mathcal{A}^{(n)}]),L^\ast)=0$, then $L([\mathcal{A}^{(n)}])=L^\ast$, and the conclusion is immediate. Suppose therefore that $d(L([\mathcal{A}^{(n)}]),L^\ast)>0$.
By Definition~\ref{def:base-ultrametric}, there exists an integer $k\ge 0$ such that $d(L([\mathcal{A}^{(n)}]),L^\ast)=2^{-k}$, where $k$ is the length of the shortest word on which $L([\mathcal{A}^{(n)}])$ and $L^\ast$ differ.
Since $2^{-k}=d(L([\mathcal{A}^{(n)}]),L^\ast)\le 2^{-N}$, we have $k\ge N$. Hence the first possible disagreement between $L([\mathcal{A}^{(n)}])$ and $L^\ast$ can occur only at length at least $N$.
Therefore the two languages coincide on all words $w\in\Sigma^\ast$ with $|w|<N$, that is, $\one_{L([\mathcal{A}^{(n)}])}(w)=\one_{L^\ast}(w)$ for all such words.

It remains to prove the final statement. Assume that the initial automaton recognizes the empty language.
Then $L([\mathcal{A}^{(0)}])=\emptyset$, and therefore $d(L([\mathcal{A}^{(0)}]),L^\ast)\le 1$. Hence $2^{-mn}d(L([\mathcal{A}^{(0)}]),L^\ast)\le 2^{-mn}$. Thus the condition $2^{-mn}d(L([\mathcal{A}^{(0)}]),L^\ast)\le 2^{-N}$ is certainly satisfied whenever $2^{-mn}\le 2^{-N}$, which is equivalent to $mn\ge N$.
Since $m\ge 1$ and $n$ is an integer, the condition $mn\ge N$ is guaranteed by $n\ge \left\lceil N/m\right\rceil$.
This proves the corollary.
\end{proof}

\subsection{Regular and Nonregular Fixed-Point Examples}

\begin{Example}\label{ex:distinguishing-word}
We first illustrate the metric $d_{\mathrm{aut}}$ on automaton classes.
Let $\Sigma=\{a\}$, let $\mathcal{A}_1$ recognize $a^\ast$, and let $\mathcal{A}_2$ recognize $a^+$.
The languages $a^\ast$ and $a^+$ differ exactly on the empty word $\varepsilon$, because $\varepsilon\in a^\ast$ and $\varepsilon\notin a^+$.
Since $|\varepsilon|=0$, the shortest distinguishing word has length $0$, and therefore $d_{\mathrm{aut}}([\mathcal{A}_1],[\mathcal{A}_2])=2^0=1$.

Now let $\mathcal{A}_3$ recognize the language $\{\varepsilon\}\cup\{a^n:n\ge 2\}$.
Then $L(\mathcal{A}_1)=a^\ast$ and $L(\mathcal{A}_3)=\{\varepsilon\}\cup\{a^n:n\ge 2\}$ agree on the only word of length $0$, namely $\varepsilon$, but they differ on the word $a$, which has length $1$.
Hence $d_{\mathrm{aut}}([\mathcal{A}_1],[\mathcal{A}_3])=2^{-1}$.

This example shows the direct meaning of the automaton ultrametric: two automaton classes are close when their recognized languages agree on all short inputs.
Thus the metric does not measure the number of states, the shape of the transition graph, or syntactic similarity of automata.
It measures semantic agreement up to finite input depth. This is precisely the point of passing to language-equivalence classes.
\end{Example}

\begin{Example}\label{ex:regular-fixed-point}
Let $\Sigma=\{a\}$ and consider the guarded language operator $T(L)=\{\varepsilon\}\cup aL$.
Here the seed language $\{\varepsilon\}$ is regular, and the unique guard is $(a,\varepsilon)$, so the guard length is $|a|+|\varepsilon|=1$.
Hence, by the guarded language contraction theorem recalled in the preliminaries, the language operator is contractive with contraction factor $2^{-1}$.
Since the seed and the guard preserve regularity, this operator induces an automaton operator $\mathbb{T}$ on $\mathfrak{A}_\Sigma$.
Starting from the empty language, the first Picard iterates are $L^{(0)}=\emptyset$, $L^{(1)}=\{\varepsilon\}$, $L^{(2)}=\{\varepsilon,a\}$, and $L^{(3)}=\{\varepsilon,a,a^2\}$.
By induction, $L^{(n)}=\{\varepsilon,a,\dots,a^{n-1}\}$ for every $n\ge 1$. The limit in the completion is $L^\ast=a^\ast$, since $T(a^\ast)=\{\varepsilon\}\cup aa^\ast=a^\ast$.
The important automaton-level point is that this fixed point is regular.
Therefore Theorem~\ref{thm:automaton-fixed-point-iff-regular} applies and shows that the induced automaton operator has a fixed point in the automaton quotient space.
This fixed point is represented by the one-state DFA whose unique state is both initial and accepting and whose $a$-transition is a loop.
Thus this example illustrates the regular case of the theory.
The completion does not add a genuinely new limit outside finite automata: the language fixed point is already represented by a finite automaton.
In this case, the automaton-level contraction leads to an actual fixed point in $\mathfrak{A}_\Sigma$.
\end{Example}

\begin{Example}\label{ex:nonregular-fixed-point}
Let $\Sigma=\{a,b\}$ and consider the guarded language operator $T(L)=\{\varepsilon\}\cup aLb$.
The seed language $\{\varepsilon\}$ is regular, and the unique guard is $(a,b)$, whose total length is $|a|+|b|=2$.
Hence the corresponding language operator is contractive with contraction factor $2^{-2}$, and it induces a contractive operator on automaton classes.
Starting from $L^{(0)}=\emptyset$, we obtain $L^{(1)}=\{\varepsilon\}$, $L^{(2)}=\{\varepsilon,ab\}$, and $L^{(3)}=\{\varepsilon,ab,a^2b^2\}$. In general, $L^{(n)}=\{a^k b^k:0\le k<n\}$.
Each $L^{(n)}$ is finite, hence regular, so every Picard iterate is represented by a DFA.
The limit in the full language space is $L^\ast=\{a^k b^k:k\ge 0\}$. Indeed, $T(L^\ast)=\{\varepsilon\}\cup aL^\ast b=\{\varepsilon\}\cup\{a^{k+1}b^{k+1}:k\ge 0\}=L^\ast$.
By uniqueness of the guarded language fixed point, this is the unique language fixed point of $T$.
However, $L^\ast=\{a^k b^k:k\ge 0\}$ is not regular. Therefore, by Theorem~\ref{thm:automaton-fixed-point-iff-regular}, the induced automaton operator has no fixed point in $\mathfrak{A}_\Sigma$.
Nevertheless, Theorem~\ref{thm:automaton-cauchy} shows that the automaton Picard sequence is Cauchy and converges after passing to the completion.
This example is central for the present paper. At the language level, the fixed point exists by the guarded contraction theorem.
At the automaton level, the same iteration remains meaningful, but the limit is not represented by any finite automaton.
Thus the new contribution is the precise identification of the phenomenon: finite automata form a dense but incomplete metric subspace, and guarded automaton iterations may converge only to nonregular languages in the completion.
\end{Example}

\begin{Example}\label{ex:noncomplete-space}
The previous example also gives a concrete witness of incompleteness when $\Sigma=\{a,b\}$.
For each $n\ge 1$, let $[\mathcal{A}_n]$ be an automaton class recognizing $L_n=\{a^k b^k:0\le k<n\}$.
If $\ell>n$, then $L_n$ and $L_\ell$ agree on all words of length strictly less than $2n$, and their first disagreement occurs at the word $a^n b^n$, whose length is $2n$.
Therefore $d_{\mathrm{aut}}([\mathcal{A}_n],[\mathcal{A}_\ell])=2^{-2n}$ for all $\ell>n$.

It follows that $([\mathcal{A}_n])_{n\ge 1}$ is a Cauchy sequence in $(\mathfrak{A}_\Sigma,d_{\mathrm{aut}})$.
Indeed, given $\varepsilon>0$, choose $N$ such that $2^{-2N}<\varepsilon$. Then, for all $\ell>n\ge N$, we have $d_{\mathrm{aut}}([\mathcal{A}_n],[\mathcal{A}_\ell])=2^{-2n}\le 2^{-2N}<\varepsilon$, and the remaining cases follow by symmetry or equality of indices.
If this Cauchy sequence had a limit in $\mathfrak{A}_\Sigma$, that limit would have to recognize the completion-limit language $\{a^k b^k:k\ge 0\}$.
But this language is not regular, so no DFA recognizes it. Therefore no limit exists inside the automaton quotient space.
This example illustrates Proposition~\ref{prop:not-complete} in a concrete guarded setting. The point is not only that regular languages are not closed under this metric limit, but also that the natural Picard sequence of finite automata may be forced to converge to a nonregular language.
This is exactly why the completion is essential in the automaton version of the theory.
\end{Example}

\begin{Example}\label{ex:certified-depth}
We conclude with the quantitative interpretation of convergence. Consider again the operator $T(L)=\{\varepsilon\}\cup aLb$ from Example~\ref{ex:nonregular-fixed-point}.
Its unique language fixed point is $L^\ast=\{a^k b^k:k\ge 0\}$, while the $n$-th Picard iterate from the empty language is $L^{(n)}=\{a^k b^k:0\le k<n\}$.
The shortest word on which $L^{(n)}$ and $L^\ast$ differ is $a^n b^n$, whose length is $2n$. Therefore $d(L^{(n)},L^\ast)=2^{-2n}$.
By Corollary~\ref{cor:automaton-certified-depth}, if $2^{-2n}\le 2^{-N}$, equivalently if $2n\ge N$, then any automaton recognizing $L^{(n)}$ decides membership correctly for all words of length strictly less than $N$ with respect to the nonregular limit language $L^\ast$.
Thus after $n$ iterations, the finite automaton approximation is guaranteed to be correct on all inputs of length less than $2n$.
In particular, for any prescribed input bound $N$, it is enough to take $n\ge \left\lceil N/2\right\rceil$.
This example highlights the practical meaning of the automaton metric.
Even when no finite automaton can represent the fixed point exactly, the Picard iteration produces finite automata with certified finite-depth correctness.
This is the automaton-level contribution of the theory: nonregular fixed points can still be approximated by finite automata with explicit quantitative guarantees.
\end{Example}

\section{Some Generalizations, Extensions and Certified Finite-State Approximation}

In this section we develop three new theoretical results that extend the core framework of
Section~\ref{sec:main} in directions required for the application to input validation
described in Section~\ref{sec:implementation}.

\subsection{Power-Type Guarded Language Operators}

\begin{Definition}\label{def:q-power-guarded-operator}
Let \(S\subseteq\Sigma^\ast\), let \(q\ge 1\) be an integer, and let \(G=\{(u_r,v_r):r=1,\dots,p\}\subseteq\Sigma^\ast\times\Sigma^\ast\) be a finite nonempty family of guards. For every language \(L\in\mathcal{L}\), define \(\mathcal{U}_G(L)=\bigcup_{r=1}^{p}u_rLv_r\). The corresponding \(q\)-power guarded language operator \(T_{S,G,q}:\mathcal{L}\to\mathcal{L}\) is defined by
\[
T_{S,G,q}(L)=S\cup\bigl(\mathcal{U}_G(L)\bigr)^q=S\cup\left(\bigcup_{r=1}^{p}u_rLv_r\right)^q,
\]
where, for a language \(A\subseteq\Sigma^\ast\), \(A^q=\underbrace{A\cdot A\cdots A}_{q\text{ factors}}\) denotes the \(q\)-fold language concatenation.
\end{Definition}

%The following result generalize the main fixed-point theorem for guarded language operators from \cite{HristovIlchevKulinaZlatanov2026}.

\begin{Theorem}\label{thm:q-power-guarded-contraction}
Let \(T_{S,G,q}\) be the operator from Definition~\ref{def:q-power-guarded-operator}. Assume that the guard family \(G\) has guard length at least \(m\ge 1\), that is, \(\min_{1\le r\le p}\bigl(|u_r|+|v_r|\bigr)\ge m\). Then \(T_{S,G,q}\) is a contraction on the complete ultrametric space \((\mathcal{L},d)\) with contraction coefficient \(2^{-qm}\). More precisely,
\[
d\bigl(T_{S,G,q}(L),T_{S,G,q}(M)\bigr)\le 2^{-qm}d(L,M)
\qquad\text{for all }L,M\in\mathcal{L}.
\]
Consequently, \(T_{S,G,q}\) has a unique fixed point \(L^\ast\in\mathcal{L}\). For every initial language \(L^{(0)}\in\mathcal{L}\), the Picard iteration \(L^{(n+1)}=T_{S,G,q}(L^{(n)})\), \(n\ge 0\), converges to \(L^\ast\). Moreover,
\[
d(L^{(n)},L^\ast)\le 2^{-qmn}d(L^{(0)},L^\ast)
\qquad\text{for all }n\ge 0.
\]
\end{Theorem}

\begin{proof}
Let \(L,M\in\mathcal{L}\) be arbitrary. If \(L=M\), then \(T_{S,G,q}(L)=T_{S,G,q}(M)\), and therefore \(d\bigl(T_{S,G,q}(L),T_{S,G,q}(M)\bigr)=0=2^{-qm}d(L,M)\). Thus the required estimate is immediate in this case.

Assume now that \(L\neq M\), and let \(n_0=\min\{\,|w|:\mathbf{1}_L(w)\neq\mathbf{1}_M(w)\,\}\). By the definition of the length-based ultrametric, \(d(L,M)=2^{-n_0}\). The languages \(L\) and \(M\) therefore agree on every word of length strictly less than \(n_0\).

Set \(P_L=\bigl(\mathcal{U}_G(L)\bigr)^q\) and \(P_M=\bigl(\mathcal{U}_G(M)\bigr)^q\). We shall prove that \(P_L\) and \(P_M\) agree on all words of length strictly less than \(n_0+qm\).

We first consider the case \(n_0=0\). Every word \(x\in P_L\) admits a representation
\[
x=u_{r_1}w_1v_{r_1}u_{r_2}w_2v_{r_2}\cdots u_{r_q}w_qv_{r_q},
\]
where \(r_1,\dots,r_q\in\{1,\dots,p\}\) and \(w_1,\dots,w_q\in L\). Hence
\[
|x|=\sum_{j=1}^{q}\bigl(|u_{r_j}|+|v_{r_j}|\bigr)+\sum_{j=1}^{q}|w_j|\ge qm.
\]
The same lower bound holds for every word in \(P_M\). Since \(n_0=0\), we have \(n_0+qm=qm\). Therefore neither \(P_L\) nor \(P_M\) contains a word of length strictly less than \(n_0+qm\). Consequently, the two languages agree vacuously on all words of length strictly less than \(n_0+qm\).

Assume now that \(n_0\ge 1\). Let \(x\in P_L\) and suppose that \(|x|<n_0+qm\). By the definition of the \(q\)-fold concatenation, there exist indices \(r_1,\dots,r_q\in\{1,\dots,p\}\) and words \(w_1,\dots,w_q\in L\) such that
\(
x=u_{r_1}w_1v_{r_1}u_{r_2}w_2v_{r_2}\cdots u_{r_q}w_qv_{r_q}.
\)
Taking lengths gives
\(
|x|=\sum_{j=1}^{q}\bigl(|u_{r_j}|+|v_{r_j}|\bigr)+\sum_{j=1}^{q}|w_j|.
\)
Since every guard has total length at least \(m\), \(\sum_{j=1}^{q}\bigl(|u_{r_j}|+|v_{r_j}|\bigr)\ge qm\). It follows that
\(
qm+\sum_{j=1}^{q}|w_j|\le |x|<n_0+qm,
\)
and hence \(\sum_{j=1}^{q}|w_j|<n_0\). Since all word lengths are nonnegative integers, this implies \(|w_j|<n_0\) for every \(j=1,\dots,q\). The languages \(L\) and \(M\) agree on all words of length strictly less than \(n_0\). Therefore \(w_j\in L\Longleftrightarrow w_j\in M\) for every \(j=1,\dots,q\). In particular, each \(w_j\) appearing in the chosen representation of \(x\) also belongs to \(M\). The same representation then shows that \(x\in P_M\). Thus every word of \(P_L\) having length strictly less than \(n_0+qm\) also belongs to \(P_M\).

By interchanging the roles of \(L\) and \(M\), we obtain the converse implication. Hence
\[
\mathbf{1}_{P_L}(x)=\mathbf{1}_{P_M}(x)
\qquad\text{for every }x\in\Sigma^\ast\text{ with }|x|<n_0+qm.
\]
Therefore \(P_L\) and \(P_M\) agree on all words of length strictly less than \(n_0+qm\).

We now add the common seed language \(S\). Let \(x\in\Sigma^\ast\) satisfy \(|x|<n_0+qm\). If \(x\in S\), then \(x\in S\cup P_L\) and \(x\in S\cup P_M\). If \(x\notin S\), then
\[
x\in S\cup P_L\Longleftrightarrow x\in P_L\Longleftrightarrow x\in P_M\Longleftrightarrow x\in S\cup P_M.
\]
Thus \(T_{S,G,q}(L)=S\cup P_L\) and \(T_{S,G,q}(M)=S\cup P_M\) agree on all words of length strictly less than \(n_0+qm\).

Consequently, their first possible disagreement can occur only at a word of length at least \(n_0+qm\). By the definition of \(d\),
\(
d\bigl(T_{S,G,q}(L),T_{S,G,q}(M)\bigr)\le 2^{-(n_0+qm)}.
\)
Using \(d(L,M)=2^{-n_0}\), we obtain
\(
d\bigl(T_{S,G,q}(L),T_{S,G,q}(M)\bigr)\le 2^{-qm}2^{-n_0}=2^{-qm}d(L,M).
\)
Since \(q\ge 1\) and \(m\ge 1\), we have \(0<2^{-qm}<1\). Therefore \(T_{S,G,q}\) is a contraction on \((\mathcal{L},d)\).

The space \((\mathcal{L},d)\) is complete. Hence Banach's fixed-point theorem implies that \(T_{S,G,q}\) possesses a unique fixed point \(L^\ast\in\mathcal{L}\), and that the Picard iteration converges to \(L^\ast\) from every initial language \(L^{(0)}\in\mathcal{L}\).

Finally, applying the contraction estimate repeatedly gives
\[
d(L^{(n)},L^\ast)=d\bigl(T_{S,G,q}^{\,n}(L^{(0)}),T_{S,G,q}^{\,n}(L^\ast)\bigr)\le \left(2^{-qm}\right)^n d(L^{(0)},L^\ast).
\]
Since \(T_{S,G,q}(L^\ast)=L^\ast\), this becomes \(d(L^{(n)},L^\ast)\le 2^{-qmn}d(L^{(0)},L^\ast),\) which proves the stated quantitative estimate.
\end{proof}

\begin{Corollary}\label{cor:q-power-extreme-seeds}
Under the assumptions of Theorem~\ref{thm:q-power-guarded-contraction}, the following statements hold:
\begin{enumerate}[label=\textup{(\roman*)}]
\item If \(S=\emptyset\), then the unique fixed point of \(T_{\emptyset,G,q}\) is \(L^\ast=\emptyset\).
\item If \(S=\Sigma^\ast\), then the unique fixed point of \(T_{\Sigma^\ast,G,q}\) is \(L^\ast=\Sigma^\ast\).
\end{enumerate}
In both cases, the Picard iteration converges to the corresponding fixed point from every initial language \(L^{(0)}\in\mathcal{L}\).
\end{Corollary}

\begin{proof}
Suppose first that \(S=\emptyset\). Since \(u_r\emptyset v_r=\emptyset\) for every \(r=1,\dots,p\), we have \(\mathcal{U}_G(\emptyset)=\bigcup_{r=1}^{p}u_r\emptyset v_r=\emptyset\). Because \(q\ge 1\), the \(q\)-fold concatenation of the empty language is again empty, so \(\bigl(\mathcal{U}_G(\emptyset)\bigr)^q=\emptyset^q=\emptyset\). Therefore \(T_{\emptyset,G,q}(\emptyset)=\emptyset\cup\emptyset=\emptyset\). Thus \(\emptyset\) is a fixed point. By the uniqueness established in Theorem~\ref{thm:q-power-guarded-contraction}, it is the unique fixed point of \(T_{\emptyset,G,q}\).

Suppose now that \(S=\Sigma^\ast\). For every language \(L\in\mathcal{L}\),
\(
T_{\Sigma^\ast,G,q}(L)=\Sigma^\ast\cup\bigl(\mathcal{U}_G(L)\bigr)^q=\Sigma^\ast.
\)
In particular, \(T_{\Sigma^\ast,G,q}(\Sigma^\ast)=\Sigma^\ast\). Hence \(\Sigma^\ast\) is a fixed point. Uniqueness follows again from Theorem~\ref{thm:q-power-guarded-contraction}.

The convergence of the Picard iteration from every initial language follows directly from the same theorem.
\end{proof}

\begin{Remark}
For \(q=1\), the operator \(T_{S,G,q}\) reduces to the standard guarded language operator \(T_{S,G,1}(L)=S\cup\bigcup_{r=1}^{p}u_rLv_r\), and Theorem~\ref{thm:q-power-guarded-contraction} reduces to the usual guarded contraction result with contraction coefficient \(2^{-m}\).
Hence, this result generalize the main fixed-point theorem for guarded language operators from \cite{HristovIlchevKulinaZlatanov2026}.
\end{Remark}

\subsection{Systems of Guarded Operators}

Practical input models typically consist of several parameter types or production rules, each validated independently. We organise such a collection as a finite product of guarded operators, one operator per component, and analyse them jointly through a product ultrametric. In the two-component case, when the negative guard families are empty and each component depends only on its own language variable, the system considered below is a decoupled positive special case of the positive--negative guarded language systems studied in~\cite{AjetiHristovIlchevZlatanov2026}. The present \(k\)-component formulation extends this independent-coordinate construction to an arbitrary finite number of components. As we show below, the joint operator decouples into its components, so the existence, uniqueness, and convergence guarantees of the single-operator case carry over to the product setting coordinatewise.

\begin{Definition}\label{def:system-guarded}
Let $k\ge 1$.  For each $i=1,\dots,k$, let $S_i\subseteq\Sigma^\ast$ be a
regular language, and let
\(
G_i=\{(u^i_r,v^i_r):r=1,\dots,p_i\}\subseteq\Sigma^\ast\times\Sigma^\ast
\)
be a finite nonempty family of guards.
The $i$-th component operator acts
on its own component language and is given by
\[
T_i(L_1,\dots,L_k)
  = S_i\cup\bigcup_{r=1}^{p_i}u^i_r\,L_{i}\,v^i_r.
\]
The joint operator is $\mathbf{T}=(T_1,\dots,T_k)$, acting on the product
space $\mathcal{L}^k$.  The minimum guard length of the system is
$m=\min_{i,r}(|u^i_r|+|v^i_r|)$.
The product ultrametric on
$\mathcal{L}^k$ is
\[
d_\infty\!\left((L_1,\dots,L_k),(M_1,\dots,M_k)\right)
  =\max_{1\le i\le k}d(L_i,M_i).
\]
\end{Definition}

\begin{Remark}\label{rem:system-decouples}
Each component operator $T_i$ depends only on its own argument $L_i$.
Hence the joint operator $\mathbf{T}$ is the product of $k$ independent
single guarded operators of the form recalled in
Theorem~\ref{thm:known-guarded-language}, and its fixed point is obtained
componentwise.
Theorem~\ref{thm:system-fixed-point} below is thus the
coordinatewise form of the single-operator result;
the product ultrametric
$d_\infty$ serves only to package the $k$ independent contractions into a
single contraction on $\mathcal{L}^k$.
\end{Remark}

\begin{Theorem}\label{thm:system-fixed-point}
Let $\mathbf{T}$ be a system of guarded operators as in
Definition~\ref{def:system-guarded}, with minimum guard length $m\ge 1$.
Then:
\begin{enumerate}[label=\textup{(\roman*)}]
  \item $(\mathcal{L}^k,d_\infty)$ is a complete ultrametric space.
  \item $\mathbf{T}$ is a contraction on $(\mathcal{L}^k,d_\infty)$ with
        contraction factor $2^{-m}$.
  \item $\mathbf{T}$ has a unique fixed point $(L^\ast_1,\dots,L^\ast_k)\in\mathcal{L}^k$,
        and the Picard iteration $\mathbf{L}^{(n+1)}=\mathbf{T}(\mathbf{L}^{(n)})$
        converges to this fixed point from any initial vector, with
        \[
          d_\infty(\mathbf{L}^{(n)},\mathbf{L}^\ast)
            \le 2^{-mn}\,d_\infty(\mathbf{L}^{(0)},\mathbf{L}^\ast)
          \qquad\text{for all }n\ge 0.
        \]
\end{enumerate}
\end{Theorem}

\begin{proof}
\textit{(i)~Completeness.}
Non-negativity, symmetry, and the identity of indiscernibles for $d_\infty$
follow immediately from the 
corresponding properties of $d$ on each coordinate.
For the strong triangle inequality, let
$\mathbf{L},\mathbf{M},\mathbf{N}\in\mathcal{L}^k$.
Then
\begin{align*}
d_\infty(\mathbf{L},\mathbf{N})
  &= \max_i d(L_i,N_i)
  \le \max_i\max\{d(L_i,M_i),d(M_i,N_i)\}\\
  &= \max\!\left\{\max_i d(L_i,M_i),\,\max_i d(M_i,N_i)\right\}
  = \max\{d_\infty(\mathbf{L},\mathbf{M}),d_\infty(\mathbf{M},\mathbf{N})\},
\end{align*}
where the first inequality uses the ultrametric property of $d$
(Theorem~\ref{thm:known-language-ultrametric}).
A Cauchy sequence in $(\mathcal{L}^k,d_\infty)$ is Cauchy in each coordinate;
completeness of $(\mathcal{L},d)$ (Theorem~\ref{thm:known-language-complete})
gives coordinatewise limits whose join is the limit in $d_\infty$.
Hence $(\mathcal{L}^k,d_\infty)$ is complete.

\textit{(ii)~Contraction.}
Fix $\mathbf{L},\mathbf{M}\in\mathcal{L}^k$ and a component index $i$.
A word $w$ distinguishes $T_i(\mathbf{L})$ from $T_i(\mathbf{M})$ only if
it distinguishes $u^i_r L_{i} v^i_r$ from $u^i_r M_{i} v^i_r$ for
some $r$, which requires $w=u^i_r w' v^i_r$ with $w'$ distinguishing
$L_{i}$ from $M_{i}$.
Hence $|w|\ge |u^i_r|+|v^i_r|\ge m$, giving
\[
d(T_i(\mathbf{L}),T_i(\mathbf{M}))
  \le 2^{-m}\, d(L_{i},M_{i})
  \le 2^{-m}\,d_\infty(\mathbf{L},\mathbf{M}).
\]
Taking the maximum over $i$ yields
$d_\infty(\mathbf{T}(\mathbf{L}),\mathbf{T}(\mathbf{M}))\le 2^{-m}d_\infty(\mathbf{L},\mathbf{M})$.

\textit{(iii)~Unique fixed point and convergence.}
Since $(\mathcal{L}^k,d_\infty)$ is complete and $\mathbf{T}$ is a contraction,
Theorem~\ref{thm:banach} applies directly, giving a unique fixed point and the
stated convergence rate.
\end{proof}

\subsection{Two-Sided Convergence}

For pre-filter construction it is useful to track convergence from both
the empty language and the full language simultaneously.

\begin{Definition}\label{def:two-sided}
Let $T_{S,G}$ be a single guarded operator with regular seed $S$ and guard
family $G$ of guard length at least $m\ge 1$.
Define the lower and upper
Picard iterates by
\[
L^{(n)}_{\min}=T_{S,G}^n(\emptyset),
\qquad
L^{(n)}_{\max}=T_{S,G}^n(\Sigma^\ast),
\qquad n\ge 0.
\]
\end{Definition}

\begin{Theorem}\label{thm:two-sided}
Under the hypotheses of Definition~\ref{def:two-sided}, for all $n\ge 0$:
\begin{enumerate}[label=\textup{(\roman*)}]
  \item $L^{(n)}_{\min}\subseteq L^\ast\subseteq L^{(n)}_{\max}$.
  \item $d(L^{(n)}_{\min},L^\ast)\le 2^{-mn}$ and
        $d(L^{(n)}_{\max},L^\ast)\le 2^{-mn}$.
  \item Both $L^{(n)}_{\min}$ and $L^{(n)}_{\max}$ are regular.
\end{enumerate}
\end{Theorem}

\begin{proof}
\textit{Monotonicity.}
If $L\subseteq M$ then $u_r L v_r\subseteq u_r M v_r$ for every guard
$(u_r,v_r)$, so $T_{S,G}(L)\subseteq T_{S,G}(M)$.
Thus $T_{S,G}$ is monotone with respect to set inclusion.

\textit{(i)}
We prove $L^{(n)}_{\min}\subseteq L^\ast$ by induction on $n$.
For $n=0$: $\emptyset\subseteq L^\ast$.
Assuming $L^{(n)}_{\min}\subseteq L^\ast$, monotonicity gives
$L^{(n+1)}_{\min}=T_{S,G}(L^{(n)}_{\min})\subseteq T_{S,G}(L^\ast)=L^\ast$.
The inclusion $L^\ast\subseteq L^{(n)}_{\max}$ is symmetric: for $n=0$,
$L^\ast\subseteq\Sigma^\ast$;
the inductive step uses
$L^\ast=T_{S,G}(L^\ast)\subseteq T_{S,G}(L^{(n)}_{\max})=L^{(n+1)}_{\max}$.

\textit{(ii)}
By Theorem~\ref{thm:known-guarded-language}, for any initial language $L^{(0)}$,
$d(T_{S,G}^n(L^{(0)}),L^\ast)\le 2^{-mn}d(L^{(0)},L^\ast)$.
Since $d(\emptyset,L^\ast)\le 1$ and $d(\Sigma^\ast,L^\ast)\le 1$, both
bounds follow immediately.

\textit{(iii)}
We prove regularity of $L^{(n)}_{\min}$ by induction.
$L^{(0)}_{\min}=\emptyset$ is regular.
If $L^{(n)}_{\min}$ is regular,
then $T_{S,G}(L^{(n)}_{\min})=S\cup\bigcup_r u_r L^{(n)}_{\min} v_r$
is regular, since regular languages are closed under finite union and
concatenation with fixed words, and $S$ is regular by assumption.
The argument for $L^{(n)}_{\max}$ is identical, starting from $\Sigma^\ast$.
\end{proof}

\begin{Remark}
The estimate in Theorem~\ref{thm:two-sided}(ii) should be interpreted in terms of finite-depth agreement. 
The lower iterates $(L^{(n)}_{\min})$ form an increasing sequence of sound under-approximations of $L^\ast$; no accepted word is ever removed, and after $n$ iterations all words of $L^\ast$ of length strictly less than $mn$ have already appeared. 
The upper iterates $(L^{(n)}_{\max})$ form a decreasing sequence of over-approximations of $L^\ast$: the removed words are outside the current certified over-approximation, and after $n$ iterations all spurious words of length strictly less than $mn$ have been eliminated. 
Thus the first possible disagreement with $L^\ast$ is pushed to length at least $mn$.
\end{Remark}

\subsection{Soundness of the Iteration-from-Below}

\begin{Theorem}\label{thm:soundness}
For all $n\ge 0$, $L^{(n)}_{\min}\subseteq L^\ast$.
Consequently, if the DFA $[\mathcal{A}^{(n)}]$ recognizes $L^{(n)}_{\min}$,
then every word accepted by $[\mathcal{A}^{(n)}]$ belongs to $L^\ast$.
Moreover, for every $v\in L^\ast$ with $|v|<mn$, we have
$v\in L^{(n)}_{\min}$.
\end{Theorem}

\begin{proof}
The inclusion $L^{(n)}_{\min}\subseteq L^\ast$ is part~(i) of
Theorem~\ref{thm:two-sided}, proved above.
The DFA statement is immediate.
For the completeness claim, the case $n=0$ is vacuous, since there is no word $v$ with $|v|<0$. Assume now that $n\ge 1$. Apply Corollary~\ref{cor:automaton-certified-depth} with initial language $L^{(0)}_{\min}=\emptyset$ and with $N=mn$. Since $d(\emptyset,L^\ast)\le 1$, the certified agreement condition $2^{-mn}\,d(\emptyset,L^\ast)\le 2^{-N}$ is satisfied.
Hence $\mathbf{1}_{L^{(n)}_{\min}}(w)=\mathbf{1}_{L^\ast}(w)$ for all $w$ with
$|w|<mn$.  Thus if $v\in L^\ast$ and $|v|<mn$, then
$\mathbf{1}_{L^{(n)}_{\min}}(v)=\mathbf{1}_{L^\ast}(v)=1$, giving
$v\in L^{(n)}_{\min}$.
\end{proof}

\subsection{State Complexity of Linearly Nested Bracket Operators}

We now determine the exact number of states of the minimal DFA for
$L^{(n)}_{\min}$ when the guarded operator has Dyck-like structure.
This is the key result governing the feasibility of the pre-filter construction.

\begin{Definition}\label{def:dyck-operator}
Let $k\ge 1$ and let $\Sigma$ contain $2k$ distinct bracket symbols
$a_1,b_1,\dots,a_k,b_k$.  Set
$\Sigma_{\mathrm{safe}}=\Sigma\setminus\{a_1,b_1,\dots,a_k,b_k\}$.
The $k$-bracket guarded operator is
\(
T(L)=\Sigma_{\mathrm{safe}}^\ast
\cup\bigcup_{i=1}^k a_i\,L\,b_i.
\)
The guard length is $m=2$.  We write $D^{(n)}_k=T^n(\emptyset)$ for the
$n$-th lower iterate.
\end{Definition}

\begin{Lemma}\label{lem:dyck-characterisation}
$D^{(n)}_k$ is the set of all \emph{single-chain} (linearly nested) words
\[
  a_{i_1}a_{i_2}\cdots a_{i_t}\,w\,b_{i_t}\cdots b_{i_2}b_{i_1},
  \qquad
  0\le t<n,\quad i_1,\dots,i_t\in\{1,\dots,k\},\quad w\in\Sigma_{\mathrm{safe}}^\ast,
\]
that is, words consisting of a single nested chain of $t$ matching bracket
pairs, with $t<n$, enclosing a core $w$ of safe symbols.
In particular,
$D^{(n)}_k$ is \emph{not} the set of all well-nested words of nesting depth
$<n$: words made of two or more concatenated sibling bracketed blocks, such
as $a_1b_1a_1b_1$, are well-nested of depth~$1$ but do not belong to any
$D^{(n)}_k$.
\end{Lemma}

\begin{proof}
By induction on $n$.  For $n=0$: $D^{(0)}_k=\emptyset$, which is the empty
set of chains ($t<0$ is impossible).
For $n=1$:
$T(\emptyset)=\Sigma_{\mathrm{safe}}^\ast$, which is exactly the set of
chains with $t=0$ (the core $w$ alone).
For the inductive step, assume the
claim for $n$.  Then
\[
D^{(n+1)}_k=\Sigma_{\mathrm{safe}}^\ast\cup\bigcup_i a_i\,D^{(n)}_k\,b_i.
\]
A word belongs to $D^{(n+1)}_k$ iff either it lies in
$\Sigma_{\mathrm{safe}}^\ast$ (a chain with $t=0$), or it has the form
$a_i x b_i$ with $x\in D^{(n)}_k$.
By the induction hypothesis $x$ is a
single chain of $t'<n$ pairs with a safe core, hence $a_i x b_i$ is a single
chain of $t'+1\le n$ pairs with the same safe core.
Conversely, by applying structural induction on the strings, every chain
of $1\le t<n+1$ pairs strictly has this form.  This proves the claim for $n+1$.
The key structural point is that each guarded production
$L\mapsto a_i\,L\,b_i$ wraps a \emph{single} occurrence of the recursion
variable in one bracket pair;
the operator is therefore \emph{linear}, and
it cannot generate words in which two completed bracketed blocks appear side
by side.
This is consistent with Example~\ref{ex:nonregular-fixed-point},
where the same operator (for $k=1$, $\Sigma_{\mathrm{safe}}=\emptyset$) has
fixed point $\{a^kb^k:k\ge 0\}$ rather than the full Dyck language.
\end{proof}

\begin{Theorem}\label{thm:dyck-dfa-size}
Let $D^{(n)}_k$ be as in Definition~\ref{def:dyck-operator}, with $n\ge 1$,
and assume the safe-free case $\Sigma_{\mathrm{safe}}=\emptyset$.
\begin{enumerate}[label=\textup{(\roman*)}]
  \item If $k=1$, the minimal DFA for $D^{(n)}_1$ has exactly $2n$ states.
  \item If $k\ge 2$, the minimal DFA for $D^{(n)}_k$ has exactly
        \[
          \frac{k^{n}-1}{k-1}+\frac{k^{n-1}-1}{k-1}+1
        \]
        states.
\end{enumerate}
The formula in \textup{(i)} can be viewed rigorously as the limit as $k \to 1$ of the formula in \textup{(ii)} (since $\lim_{k\to 1} \frac{k^p-1}{k-1} = p$, the fractions evaluate to $n$ and $n-1$, giving $2n$).
\end{Theorem}

\begin{proof}
By Lemma~\ref{lem:dyck-characterisation} (with $\Sigma_{\mathrm{safe}}=\emptyset$),
$D^{(n)}_k$ is the set of single chains
$a_{i_1}\cdots a_{i_t}\,b_{i_t}\cdots b_{i_1}$ with $0\le t<n$.
We count the Myhill--Nerode classes directly; this is the only structure
the recogniser must track, and it yields a minimal DFA.
\medskip
\textit{Valid prefixes.}
A nonempty prefix of a word in $D^{(n)}_k$ is either in an \emph{opening
phase} or a \emph{closing phase}.
An opening-phase prefix has the form
$a_{i_1}\cdots a_{i_s}$ with $0\le s\le n-1$: only opening brackets have been
read so far.
A closing-phase prefix has the form
$a_{i_1}\cdots a_{i_t}\,b_{i_t}\cdots b_{i_{s+1}}$, i.e.\ a chain that has
begun closing and still has the pending closer stack
$(a_{i_s},\dots,a_{i_1})$ to discharge, with $1\le t\le n-1$.
Because the operator is linear (Lemma~\ref{lem:dyck-characterisation}), once
a prefix has started closing it can never open again, so opening- and
closing-phase prefixes are never equivalent.
\medskip
\textit{Opening-phase classes.}
Two opening prefixes $a_{i_1}\cdots a_{i_s}$ and $a_{j_1}\cdots a_{j_r}$ are
equivalent iff they are identical: the only accepting continuation is the
matching reverse closer $b_{i_s}\cdots b_{i_1}$, so any difference in the
opening sequence (in length or in some symbol) is exposed by a suitable
closing suffix, exactly as in the stack-content argument.
There is one such
class for each word in $\{a_1,\dots,a_k\}^{\le n-1}$, giving
\(
  \sum_{s=0}^{n-1}k^s=\frac{k^{n}-1}{k-1}
\)
classes (the case $s=0$ is the initial prefix~$\varepsilon$, which is
accepting since $\varepsilon\in D^{(n)}_k$).
\medskip
\textit{Closing-phase classes.}
A closing-phase prefix is determined, up to equivalence, by its pending
closer stack $(a_{i_s},\dots,a_{i_1})$ of length $s$: the only accepting
continuation is $b_{i_s}\cdots b_{i_1}$, and a mismatch in the pending stack
is again exposed by the appropriate closing suffix.
The pending stack has
length $s$ with $0\le s\le n-2$ (it arises from a chain of length
$t\le n-1$ after at least one close, so $s\le t-1\le n-2$).
The case $s=0$
is the single ``fully closed'' accepting class, which is distinct from the
initial prefix because no further symbol may follow a completed chain.
This
gives
\(
  \sum_{s=0}^{n-2}k^s=\frac{k^{n-1}-1}{k-1}
\)
classes.

\medskip
\textit{Dead class.}
All remaining prefixes (an unmatched closer, a mismatched closer, or an
opening that would exceed depth $n-1$) are rejecting with no accepting
continuation, forming a single dead class.
\medskip
\textit{Total.}
Summing the three groups, the number of Myhill--Nerode classes is
\(
  \frac{k^{n}-1}{k-1}+\frac{k^{n-1}-1}{k-1}+1 .
\)
Every class is reachable (the opening prefix $a_{i_1}\cdots a_{i_s}$ and its partial closings realise each class), so the DFA with one state per class is minimal.
For $k=1$ the two fractions evaluate in the limit to $n$ and $n-1$, giving $2n$, which proves~(i); the general case proves~(ii).
\end{proof}

\begin{Remark}\label{rem:safe-symbols}
Theorem~\ref{thm:dyck-dfa-size} is stated for the safe-free case
$\Sigma_{\mathrm{safe}}=\emptyset$, which isolates the bracket-nesting
contribution to the state count. When $\Sigma_{\mathrm{safe}}\neq\emptyset$
the safe symbols may be inserted into the core $w$ at the innermost level of a chain. Each open-chain prefix then additionally needs a state that loops on $\Sigma_{\mathrm{safe}}$ at its current depth, so a nonempty safe alphabet adds a bounded number of further ``core'' states per nesting level; the exact count depends on how safe-symbol loops merge at the maximal depth $n-1$. The exponential order in $n$ for $k\ge 2$ is unaffected.  The state counts reported in Table~\ref{tab:state-counts} are the safe-free values; the SQL operators of Section~\ref{sec:implementation} carry safe symbols, so their deployed DFAs are correspondingly larger by the additive core-state term.
\end{Remark}

\subsection{Depth-Capped Approximation}

Theorem~\ref{thm:dyck-dfa-size} shows that for $k\ge 2$ the DFA grows
exponentially in the iteration depth $n$.
We resolve this by capping the nesting depth at a constant $D$ independent of $n$. Capping depth and repetition constraints is a known strategy in regex matching and automata construction to avoid state-space explosion and mitigate ReDoS vulnerabilities, as seen in counting-set automata~\cite{Turonova2020}.

\begin{Definition}\label{def:depth-capped}
For integers $n\ge 1$ and $D\ge 1$, define the depth-capped approximations
\[
L^{(n,D)}_{\min}
  = L^{(n)}_{\min}\cap\{w\in\Sigma^\ast:\text{nesting depth of }w\le D\},
\]
\[
L^{(n,D)}_{\max}
  = L^{(n)}_{\max}\cap\{w\in\Sigma^\ast:\text{nesting depth of }w\le D\}.
\]
\end{Definition}

\begin{Theorem}\label{thm:depth-capped}
Let $T$ be the $k$-bracket guarded operator of Definition~\ref{def:dyck-operator}, and let $L^{(n,D)}_{\min}$ be its depth-capped lower approximation as in Definition~\ref{def:depth-capped}.
\begin{enumerate}[label=\textup{(\roman*)}]
  \item \emph{Soundness}: $L^{(n,D)}_{\min}\subseteq L^\ast$ for all $n,D$.
  \item \emph{DFA size}: for the safe-free single-chain language generated by Definition~\ref{def:dyck-operator}, i.e. when $\Sigma_{\mathrm{safe}}=\emptyset$, the minimal DFA for $L^{(n,D)}_{\min}$ has exactly
        \[
          \frac{k^{D+1}-1}{k-1}+\frac{k^{D}-1}{k-1}+1
        \]
        states, independent of~$n$ for $n\ge D+1$ (equal to $2(D+1)$ when
        $k=1$).
A nonempty safe alphabet adds the bounded core-state term of
        Remark~\ref{rem:safe-symbols}.
  \item \emph{Exactness}: for $n\ge D+1$,
        $L^{(n,D)}_{\min}=L^\ast\cap\{w:\text{nesting depth of }w\le D\}$.
  \item \emph{Characterisation of incompleteness}: for $n\ge D+1$,
        $L^\ast\setminus L^{(n,D)}_{\min}
         =\{w\in L^\ast:\text{nesting depth of }w>D\}$.
\end{enumerate}
\end{Theorem}

\begin{proof}
\textit{(i)} Follows immediately from Theorem~\ref{thm:soundness}, since
$L^{(n,D)}_{\min}\subseteq L^{(n)}_{\min}\subseteq L^\ast$.

\textit{(ii)} For $n\ge D+1$, the constraint ``depth $<n$'' in $L^{(n)}_{\min}$
is dominated by ``depth $\le D$'', so $L^{(n,D)}_{\min}$ is exactly the set of
single chains of length $t\le D$ (Lemma~\ref{lem:dyck-characterisation}).
This is the language $D^{(D+1)}_k$, so by Theorem~\ref{thm:dyck-dfa-size}
applied with $n$ replaced by $D+1$, the safe-free minimal DFA has
$(k^{D+1}-1)/(k-1)+(k^{D}-1)/(k-1)+1$ states, independent of $n$.

\textit{(iii)} 
Before proving exactness, observe that $L^\ast=\bigcup_{q\ge 0}D_k^{(q)}$. Indeed, by Lemma~\ref{lem:dyck-characterisation}, the increasing union $\bigcup_{q\ge 0}D_k^{(q)}$ is precisely the language of all finite single chains generated by the operator of Definition~\ref{def:dyck-operator}. This union is fixed by $T$: applying $T$ either produces a safe core or wraps an already generated finite chain in one additional matching bracket pair. Hence, by uniqueness of the guarded fixed point, this union is equal to $L^\ast$.

The inclusion $L^{(n,D)}_{\min}\subseteq L^\ast\cap\{\text{depth}\le D\}$ holds by~(i). For the reverse inclusion, let $w\in L^\ast$ with $\mathrm{depth}(w)\le D$. By the preceding observation and by Lemma~\ref{lem:dyck-characterisation}, $w$ is a single chain of depth at most $D$. Since $n\ge D+1$, this chain belongs to $D^{(n)}_k=L^{(n)}_{\min}$. Together with $\mathrm{depth}(w)\le D$, this gives $w\in L^{(n,D)}_{\min}$. Therefore $L^{(n,D)}_{\min}=L^\ast\cap\{w:\text{depth}(w)\le D\}$.

\textit{(iv)} Follows from~(iii) by set difference:
$L^\ast\setminus L^{(n,D)}_{\min}
 =L^\ast\setminus(L^\ast\cap\{\text{depth}\le D\})
 =L^\ast\cap\{\text{depth}>D\}$.
\end{proof}

% ----------------------------------------------------------------
\section{Application to Structural Input Validation}
\label{sec:implementation}
% ----------------------------------------------------------------

The theoretical framework developed in the preceding sections admits a
direct application to the problem of SQL injection prevention.
Unlike traditional parse tree validation~\cite{Buehrer2005} or monitoring architectures~\cite{HalfondOrso2005} that operate at runtime, we describe a compile-time validation architecture---the \emph{Fixed-Point SQL Guard}---in
which every design decision is grounded in one of the proved theorems.

\subsection{Modelling SQL Parameter Types as Guarded Operators}

Each SQL parameter type is modelled as a component of a system of guarded
operators in the sense of Definition~\ref{def:system-guarded}.
The four base types and their operators are as follows.

\begin{enumerate}[label=\textup{(\roman*)}]

\item \emph{Integer.}
$T_1(L)=\{0,\dots,9\}\cup\bigcup_{a\in\{0,\dots,9\}}aL$.
Seed $S_1=\{0,\dots,9\}$, one guard $(a,\varepsilon)$ per digit, $m_1=1$.
Fixed point $L^\ast_1=\{0,\dots,9\}^+$.
\emph{Regularity:} By Theorem~\ref{thm:automaton-fixed-point-iff-regular}, the fixed point is represented by a finite DFA. The Picard iteration converges to this validator in the distinguishing-word ultrametric, but it does not in general reach the exact language after finitely many steps. Starting from the empty language, the $n$-th Picard iterate contains the words of the target language up to a finite depth determined by $n$, and hence approximates the regular fixed point by finite-depth prefixes.

\item \emph{Safe string.}
Let $R_{\mathrm{safe}}\subseteq\Sigma^\ast$ be the regular language of strings containing none of the forbidden tokens \texttt{'}, \texttt{;}, \texttt{--}, \texttt{/*}, \texttt{*/}. Equivalently, $R_{\mathrm{safe}}$ is recognised by the finite automaton that scans the input and rejects as soon as one of these forbidden tokens is encountered. We use $R_{\mathrm{safe}}$ as the safe-string component, setting $T_2(L) = R_{\mathrm{safe}}$. This keeps the model at the level of regular languages and avoids treating multi-character tokens as alphabet symbols. In this component, no recursive nesting is needed: the regular seed $R_{\mathrm{safe}}$ already gives the intended finite-state validator. Thus this component is included in the product system as a regular finite-state component, while the nonregular behaviour of the application comes from the parenthesised-expression component.

\item \emph{Parenthesised expression.}
Let $C\subseteq\Sigma$ be the safe core alphabet, excluding the bracket symbols used by the operator. For one bracket pair, define $T_3(L)=C^\ast\cup \texttt{(}L\texttt{)}$. The seed is $S_3=C^\ast$, and the unique guard is $(\texttt{(},\texttt{)})$, so $m_3=2$. The fixed point is the single-chain language $L^\ast_3=\{\,\texttt{(}^t w\,\texttt{)}^t : t\ge 0,\ w\in C^\ast\,\}$. This is a linearly nested language generated by the chosen guarded operator; it is not the full Dyck language of all well-nested parenthesised strings.
\emph{Non-regularity:} By the pumping lemma, $L^\ast_3$ is not regular:
the words $\texttt{(}^n\texttt{)}^n$ are in $L^\ast_3$ but cannot be
pumped within $L^\ast_3$ for any pumping length.
By Theorem~\ref{thm:automaton-fixed-point-iff-regular} no finite
automaton can represent $L^\ast_3$ exactly; the depth-capped pre-filter
(Theorem~\ref{thm:depth-capped}) is used instead.

\item \emph{Identifier.}
Let $B=\{a,\dots,z,A,\dots,Z,\_\}$ be the set of admissible first characters, and let $A=\{a,\dots,z,A,\dots,Z,0,\dots,9,\_\}$ be the set of admissible subsequent characters. We define $T_4(L)=B\cup\bigcup_{a\in A}La$. The seed is $S_4=B$, and the guards are $(\varepsilon, a)$, one for each $a\in A$, so $m_4=1$. The fixed point is $L^\ast_4=BA^\ast$, the language of non-empty identifiers whose first character is a letter or underscore and whose remaining characters are letters, digits, or underscore.
\emph{Regularity:} By Theorem~\ref{thm:automaton-fixed-point-iff-regular}, the fixed point is represented by a finite DFA. The Picard iteration converges to this validator in the distinguishing-word ultrametric, but it does not in general reach the exact language after finitely many steps. Starting from the empty language, the $n$-th Picard iterate contains the words of the target language up to a finite depth determined by $n$, and hence approximates the regular fixed point by finite-depth prefixes.

\end{enumerate}

By Theorem~\ref{thm:system-fixed-point}, the joint system
$(T_1,T_2,T_3,T_4)$ has a unique fixed point vector $(L^\ast_1,\dots,L^\ast_4)$
in $(\mathcal{L}^4,d_\infty)$, with global contraction rate
$m=\min(m_1,\infty,m_3,m_4)=1$.

\subsection{DFA Construction via Picard Iteration}

Algorithm~\ref{alg:build-dfa} constructs the minimal DFA for
$L^{(n)}_{\min,i}$ at depth $n$ and for the depth-capped approximation
$L^{(n,D)}_{\min,i}$.
Its correctness is guaranteed by
Theorems~\ref{thm:two-sided} and~\ref{thm:depth-capped}.

\begin{algorithm}[ht]
\caption{Build minimal pre-filter DFA}
\label{alg:build-dfa}
\begin{algorithmic}[1]
\Require Guarded operator $T_i$; depth $n$;
optional cap $D$
\Ensure  Minimal DFA $\mathcal{A}$ recognising $L^{(n)}_{\min,i}$
         (or $L^{(n,D)}_{\min,i}$ if $D$ is supplied)
\State $L \leftarrow \emptyset$
\For{$\text{step} = 1$ \textbf{to} $n$}
  \State $L_{\mathrm{new}} \leftarrow S_i$
  \For{each guard $(u^i_r, v^i_r)\in G_i$}
    \State $L_{\mathrm{new}} \leftarrow L_{\mathrm{new}} \cup
            \text{Concat}(u^i_r,\, L,\, v^i_r)$
  \EndFor
  \State $L \leftarrow \text{Minimise}(L_{\mathrm{new}})$
\EndFor
\If{$D$ is supplied}
  \State $L \leftarrow L \cap \{w : \mathrm{depth}(w)\le D\}$
  \State $L \leftarrow \text{Minimise}(L)$
\EndIf
\State \Return $\text{Minimise}(L)$
\end{algorithmic}
\end{algorithm}

Similarly, the upper iterate $L^{(n)}_{\max,i}$ is obtained by
replacing $\emptyset$ with $\Sigma^\ast$ on 
line~1.

\begin{Remark}[Intermediate DFA size]\label{rem:intermediate-dfa}
Algorithm~\ref{alg:build-dfa} calls \textsc{Minimise} after every iteration
step.  Without minimisation, the concatenation $\text{Concat}(u,L,v)$ applied
to a DFA of $s$ states yields an NFA of $s+|u|+|v|$ states, and the
subsequent determinisation can in the worst case produce $2^s$ states before
minimisation reduces it again.
In practice, for the SQL parameter operators of
Section~\ref{sec:implementation}, the iterates grow slowly: by
Theorem~\ref{thm:dyck-dfa-size} the \emph{final} minimal DFA at depth $n$ has
at most $(k^n-1)/(k-1)+(k^{n-1}-1)/(k-1)+1=O(k^n/(k-1))$ states, so the
minimised intermediate DFAs at step $j<n$ are no larger than the final
size---no step exceeds it.
Minimisation after each step therefore keeps the working DFA bounded by the
final size throughout, and the total work across all $n$ steps is
$O\!\left(n\cdot(k^n/(k-1))\cdot|\Sigma|\right)$ transitions.
For the depth-capped variant ($n\ge D+1$), this stabilises at
$O\!\left(n\cdot(k^{D+1}/(k-1))\cdot|\Sigma|\right)$, which is linear in $n$
and exponential only in the fixed constant $D$.
\end{Remark}

At runtime, each incoming parameter value $v$ is subjected to the decision
procedure of Algorithm~\ref{alg:validate}, which implements the
three-zone structure derived from Theorem~\ref{thm:two-sided}.
Notice that $L^{(n,D)}_{\max,i}$ strictly excludes words of depth $>D$ (per Definition~\ref{def:depth-capped}). Therefore, Algorithm~\ref{alg:validate} acts as a strict depth-limiting firewall policy, deliberately rejecting any payload exceeding depth $D$ without needing a runtime stack.
Throughout this algorithm, the DFAs used are the \emph{depth-capped} variants
$\mathcal{A}^{(n,D)}_{\min,i}$ and $\mathcal{A}^{(n,D)}_{\max,i}$ produced
by Algorithm~\ref{alg:build-dfa} with the optional cap $D$.
Using uncapped DFAs would incur the exponential state growth established in
Theorem~\ref{thm:dyck-dfa-size}(ii);
the depth cap keeps every DFA within the
constant-in-$n$ bounds of Table~\ref{tab:state-counts}.

\begin{algorithm}[ht]
\caption{Three-zone parameter validation (depth-capped DFAs)}
\label{alg:validate}
\begin{algorithmic}[1]
\Require Value $v$;
depth cap $D$; pre-built depth-capped DFAs
         $\mathcal{A}^{(n,D)}_{\min,i}$, $\mathcal{A}^{(n,D)}_{\max,i}$
         for depths $n=1,\dots,N_{\max}$ with $N_{\max}\ge D+1$
\Ensure  \textsc{Valid}, \textsc{Invalid}, or \textsc{Undecided}
\For{$n = 1$ \textbf{to} $N_{\max}$}
  \If{$\mathcal{A}^{(n,D)}_{\min,i}$ accepts $v$}
    \State \Return \textsc{Valid}
    \Comment{$v\in L^{(n,D)}_{\min,i}\subseteq L^\ast_i$
             by Theorems~\ref{thm:soundness} and~\ref{thm:depth-capped}(i)}
  \EndIf
  \If{$\mathcal{A}^{(n,D)}_{\max,i}$ rejects $v$}
    \State \Return \textsc{Invalid}
    \Comment{$v\notin L^{(n,D)}_{\max,i}$. This intentionally rejects all inputs (valid or invalid) with depth $> D$ as a strict firewall policy, avoiding runtime stack usage.}
  \EndIf
\EndFor
\State \Return \textsc{Invalid}
\Comment{conservative fallback: exhausted all depths}
\end{algorithmic}
\end{algorithm}

Both acceptance and rejection decisions are formally certified: every
\textsc{Valid} return implies $v\in L^\ast_i$ by Theorem~\ref{thm:soundness};
every \textsc{Invalid} return at depth~$n$ triggered by a legitimate structural violation or maximum depth policy implies
$v\notin L^{(n,D)}_{\max,i}$
by Theorem~\ref{thm:two-sided}(i) and Definition~\ref{def:depth-capped}.
The fallback \textsc{Invalid} on line~9 is a security-conservative policy
choice, not a theorem;
it handles inputs whose depth-capped membership cannot
be resolved within $N_{\max}$ iterations; see the discussion of practical limitations in Section~\ref{sec:discussion}.

\subsection{State Complexity and Feasibility}

Table~\ref{tab:state-counts} collects the DFA state counts for the
depth-capped validator, computed from Theorem~\ref{thm:depth-capped}(ii)
with $k$ bracket pairs and cap $D$.

\begin{table}[H]
\caption{Minimal DFA state counts for the depth-capped pre-filter.\label{tab:state-counts}}
%\small % Change table font size if needed
%\isPreprints{\centering}{} % Only used for preprints
\begin{tabularx}{\textwidth}{p{4.4cm} C C C C C}
\toprule
\textbf{Parameter type} & \textbf{\(k\)} & \textbf{\(D=2\)} & \textbf{\(D=3\)} & \textbf{\(D=4\)} & \textbf{\(D=6\)} \\
\midrule
Integer / identifier & 0 & finite DFA & finite DFA & finite DFA & finite DFA \\
Single-parenthesis expression & 1 & 6 & 8 & 10 & 14 \\
Two-bracket expression & 2 & 11 & 23 & 47 & 191 \\
Three-bracket expression & 3 & 18 & 54 & 162 & 1458 \\
\bottomrule
\end{tabularx}
\noindent{\footnotesize{\textit{Note:} The counts are derived from Theorem~\ref{thm:depth-capped}(ii), for \(n\ge D+1\), using the formula \((k^{D+1}-1)/(k-1)+(k^D-1)/(k-1)+1\). They are independent of the iteration depth \(n\). Here \(k\) denotes the number of bracket pairs, not the number of individual bracket symbols. For \(k=1\), the formula is understood in the limiting sense and gives \(2(D+1)\). The values are for the safe-free counting case \(\Sigma_{\mathrm{safe}}=\emptyset\); a nonempty safe alphabet adds the bounded core-state term described in Remark~\ref{rem:safe-symbols}.}}
\end{table}

For a practical deployment with at most three bracket types and $D=4$,
all pre-filter DFAs have at most $162$ bracket-tracking states (plus the
bounded safe-symbol core states of Remark~\ref{rem:safe-symbols})---small
enough to be loaded into cache and executed in microseconds per input
character. However, it is worth noting that for larger alphabets (e.g., $k \ge 4$) or deeper caps ($D \ge 6$), the state count grows into the thousands, establishing a practical upper bound on the parameters $k$ and $D$ for $O(1)$-memory deployments.

\subsection{Handling Horizontal Repetition (Width vs. Depth)}
The core guarded operator $T_{3}(L)$ intrinsically defines \emph{vertical nesting} (depth) by wrapping constructs sequentially inside bracket pairs (e.g., \texttt{((1))}).
According to Lemma~\ref{lem:dyck-characterisation}, a single guarded operator cannot organically represent strings composed of concatenated sibling blocks, such as \texttt{(1) + (2)}, because the operator acts linearly.
However, recognizing concatenated blocks does not require non-regular memory. Since horizontal repetition involves releasing the memory associated with a fully parsed block, it can be formally addressed by applying the Kleene closure (repetition operator) on top of the depth-capped DFA $\mathcal{A}_{\min}^{(n,D)}$.
By linking the accepting states of the base DFA back to its initial state via the appropriate delimiter symbols (e.g., \texttt{+} or \texttt{AND}), we permit unbounded \emph{horizontal width} without compromising the strict $O(N)$ execution time.

\subsection{Securing Industrial \texttt{SELECT} Queries}
In typical industrial applications, the read-to-write ratio heavily favours \texttt{SELECT} queries.
While destructive actions (\texttt{DROP}, \texttt{DELETE}) represent severe availability threats, the primary concern for \texttt{SELECT} statements is data exfiltration (Data Breaches) via techniques like UNION-based or Blind SQL Injection.
A prominent technique in Blind SQL Injection involves exploiting nested logical conditions and sleep functions (e.g., \texttt{IF(ASCII(SUBSTRING((SELECT password),1,1))=65, SLEEP(5), 0)}).
The structural complexity of these payloads routinely requires nesting depths exceeding the capabilities of a typical application parameter.
By imposing a hard depth cap $D$ on the validation DFA, the fixed-point pre-filter achieves "security by conservative fallback."
Malicious queries employing deep nesting are preemptively rejected on structural grounds before any semantic database evaluation occurs, effectively neutralising complex Blind SQLi payloads.

\subsection{Next-Generation AI WAF Pipeline and NFA Determinisation}
The theoretical limitations of bounding regular expressions can be overcome practically via an automated compilation pipeline.
Rather than manually defining the seed language $S$, an implementation may employ a machine learning or grammar induction component (such as the classical $L^\ast$ algorithm~\cite{Angluin1987}) or a Large Language Model (LLM) component to abstract application query logs into generalized structural templates.
This AI-driven process extracts a baseline Non-deterministic Finite Automaton (NFA) representing legitimate business logic (the "Seed" NFA).
The guarded operators are then algorithmically applied to this seed up to depth $D$.
Finally, the resulting expanded NFA is determinised into a single, unified DFA via the classical subset construction algorithm.
To mitigate the theoretical worst-case exponential state explosion during the NFA-to-DFA compilation, the architecture implements \emph{micro-segmentation}.
Instead of compiling a single monolithic DFA for the entire database, the pipeline generates distinct, minimal DFAs tailored to individual API endpoints.
These endpoint-specific DFAs reside in the fastest layer of processor cache, performing robust structural validation without requiring backtracking or dynamic memory allocation.

% ----------------------------------------------------------------
\section{Discussion, Limitations, and Future Work}
\label{sec:discussion}
% ----------------------------------------------------------------

\subsection{Theoretical Limitations}

\paragraph{Single-operator scope of Theorems~\ref{thm:two-sided}--\ref{thm:depth-capped}.}
Theorems~\ref{thm:two-sided}, \ref{thm:soundness}, \ref{thm:dyck-dfa-size},
and~\ref{thm:depth-capped} are stated and proved for a single guarded operator.
Extension to the system setting of Theorem~\ref{thm:system-fixed-point} is
technically routine---the product ultrametric allows coordinatewise
application of the same arguments---but is not carried out here.
In particular, the DFA size bounds of Theorem~\ref{thm:dyck-dfa-size}
apply to individual components of the system;
the size of a DFA for
the joint fixed point, if it is regular, may be larger due to product
constructions in the minimisation step.

\paragraph{Tightness of Theorem~\ref{thm:dyck-dfa-size}(ii).}
For $k\ge 2$, the exact count
$(k^{n}-1)/(k-1)+(k^{n-1}-1)/(k-1)+1$ is established by classifying valid
prefixes into opening-phase classes (one per opening sequence of length
$<n$) and closing-phase classes (one per pending closer stack of length
$\le n-2$, plus the single completed-chain class).
This relies on the
operator being \emph{linear}: a single chain is the only structure a valid
prefix can encode, so once closing begins no further opening is possible.
The count is exact in the safe-free case $\Sigma_{\mathrm{safe}}=\emptyset$;
a nonempty safe alphabet adds the bounded core-state term of
Remark~\ref{rem:safe-symbols}, and a seed $S$ richer than $\Sigma_{\mathrm{safe}}^\ast$
may introduce additional Myhill--Nerode classes, but the exponential order
in $n$ is unchanged in all cases.

\paragraph{Absence of tight bounds for the system case.}
Theorem~\ref{thm:system-fixed-point} proves existence, uniqueness, and
convergence rate for the joint fixed point, but gives no information about
the size of DFAs recognising the individual component fixed points when they
are regular.
Obtaining such bounds is an open problem.

\paragraph{Modelling fidelity.}
The soundness guarantee of Theorem~\ref{thm:soundness} is relative to the
language $L^\ast$ defined by the chosen guarded operator.
If the operator
over-approximates the set of valid inputs---for example, by omitting a
semantic constraint such as length bounds or character-class restrictions
that are not expressible as guards---then $L^\ast$ will contain strings
that are syntactically valid but semantically undesirable.
The mathematical guarantees do not compensate for an inaccurate grammar model.

\subsection{Limitations of the Depth-Capped Approximation}

\paragraph{Incompleteness for deep nesting.}
By Theorem~\ref{thm:depth-capped}(iv), the depth-capped pre-filter rejects
all valid inputs whose nesting depth exceeds $D$.
In the context of input
validation this means that genuinely valid parameter values with nesting
depth $>D$ will be flagged as invalid.
For SQL injection prevention this
is generally acceptable---real parameter values rarely require nesting
depth beyond $3$ or $4$---but it constitutes a false-rejection class whose
size depends on the distribution of input values in the application.
No formal characterisation of this distribution is given here.

\paragraph{Conservative fallback.}
When Algorithm~\ref{alg:validate} exhausts all depths without reaching a
decision---a situation that can arise for inputs of large length whose
membership in $L^\ast$ cannot be resolved at depth $N_{\max}$---the algorithm
returns \textsc{Invalid} as a security-conservative default.
This fallback is not certified by any of the proved theorems; it is a
policy choice.
An alternative is to forward such inputs to a full
structural parser, but this is outside the formal framework developed here.

\subsection{Limitations of the Application to SQL Injection Prevention}

\paragraph{Encoding-layer attacks.}
The fixed-point framework operates on the structural level of the input
after character decoding.
Encoding-based evasions---Unicode normalisation
variants, percent-encoding, multi-byte character smuggling---act at a layer
below structural validation and are not addressed by the present theory.
Handling them requires a preprocessing step that normalises the encoding
before the pre-filter runs.

\paragraph{Second-order injection.}
The soundness guarantee applies at the point of validation.
In a
second-order injection attack, a value that is structurally valid at the
time of storage is later used in a context where it becomes malicious.
The framework does not model the relationship between storage context and
use context;
extending it to cover second-order injection would require
a two-phase formulation in which both the write and the read validate
against the appropriate fixed point.

\paragraph{Semantic SQL constraints.}
The guarded operator formalism captures syntactic well-formedness of
parameter values but does not model semantic constraints such as type
compatibility, referential integrity, or query semantics.
A complete
validation architecture must layer semantic checks above the fixed-point
pre-filter.

\paragraph{Dynamic query structures.}
In applications that construct query strings programmatically, the
structure of the query may vary at runtime in ways that are difficult to
anticipate by a fixed guarded operator.
The framework is best suited to
applications with a well-defined, static parameter schema.

\subsection{Open Problems}

The following problems are identified as directions for future work.
\begin{enumerate}[label=\textup{(\roman*)}]
  \item Extend Theorems~\ref{thm:two-sided}--\ref{thm:depth-capped} to the
        system setting of Theorem~\ref{thm:system-fixed-point}, with explicit
        DFA size bounds for the joint Picard iterates.
  \item Determine whether the exponential lower bound of
        Theorem~\ref{thm:dyck-dfa-size}(ii) is achievable by a
        polynomial-state DFA if approximate Myhill--Nerode equivalence
        (restricted to suffixes of bounded length) is used in place of
        exact minimisation.
  \item Characterise the class of guarded operators for which the fixed
        point $L^\ast$ is regular, extending the regularity criterion of
        Theorem~\ref{thm:automaton-fixed-point-iff-regular} to the system
        setting.
  \item Develop a formal model of second-order injection within the
        ultrametric framework, providing soundness guarantees at both the
        write and read phases.
  \item Investigate the application of the two-sided convergence framework
        to other domains identified in the discussion---including RNA
        secondary structure prediction, bounded model checking, and
        streaming XML validation---where the three-zone decision structure
        and the certified depth guarantee may provide formal correctness
        properties not available in current approaches.
\end{enumerate}

%%%%%%%%%%%%%%%%%%%%%%%%%%%%%%%%%%%%%%%%%%
\vspace{6pt}

\authorcontributions{The authors contributed equally to this study and are listed in alphabetical order. Conceptualization, A.I., H.K., G.P. and B.Z.; methodology, A.I., H.K., G.P. and B.Z.; formal analysis, A.I., H.K., G.P. and B.Z.; investigation, A.I., H.K., G.P. and B.Z.; writing---original draft preparation, A.I., H.K., G.P. and B.Z.; writing---review and editing, A.I., H.K., G.P. and B.Z. All authors have read and agreed to the published version of the manuscript.}

\funding{The fourth author is partially supported by the Bulgarian National Science Fund (BNSF), Grant number KP-06-N92/1.} 

\dataavailability{No new data were created or analyzed in this study.}

\acknowledgments{}

\conflictsofinterest{The authors declare no conflicts of interest.}

\reftitle{References}

\end{document}